\documentclass[a4paper]{amsart}

\usepackage{amsmath,amsfonts,amssymb,amsthm,graphicx,cite,mathtools}
\usepackage{dsfont}
\usepackage{amscd}
\usepackage{mathrsfs}
\usepackage[pdfstartview=FitV]{hyperref}

\newtheorem{theorem}{Theorem}[section]

\newtheorem{corollary}[theorem]{Corollary}
\newtheorem{lemma}{Lemma}[section]

\theoremstyle{remark}

\newcommand{\e}{\mathrm{e}}
\newcommand{\imag}{\mathrm{i}}
\newcommand{\gtr}{>}
\newcommand{\less}{<}
\newcommand{\abs}[1]{\lvert #1 \rvert}
\newcommand{\bs}{\boldsymbol}
\newcommand{\tN}{\tilde{N}}
\newcommand{\tM}{\tilde{M}}
\newcommand{\tx}{\tilde{x}}
\newcommand{\ty}{\tilde{y}}
\newcommand{\Z}{{\mathbb Z}}
\newcommand{\C}{{\mathbb C}}
\newcommand{\R}{{\mathbb R}}
\newcommand{\cA}{{\mathcal A}}
\newcommand{\cE}{{\mathcal E}}
\newcommand{\cH}{{\mathcal H}}
\newcommand{\cN}{{\mathcal{N}}}
\newcommand{\cV}{{\mathcal{V}}}

\newcommand{\half}{\mbox{$\frac{1}{2}$}}
\newcommand{\quarter}{\mbox{$\frac{1}{4}$}}
\newcommand{\ve}{\varepsilon}

\begin{document}

\title[Source identities for the deformed van Diejen model]{Source identities and kernel functions for the deformed Koornwinder-van Diejen models}

\author{Farrokh Atai}
\address{Department of Mathematics, Kobe University, Rokko, Kobe 657-8501, Japan}
\email{farrokh@math.kobe-u.ac.jp}

\date{\today}

\begin{abstract}
We consider generalizations of the $BC$-type relativistic Calogero-Moser-Sutherland models, comprising of the rational, trigonometric, hyperbolic, and elliptic cases, due to Koornwinder and van Diejen, and construct an explicit eigenfunction for these generalizations. In special cases, we find the various kernel function identities, and also a Chalykh-Feigin-Sergeev-Veselov type deformation of these operators and their corresponding kernel functions, which generalize the known kernel functions for the Koornwinder-van Diejen models.
\newline
\begin{flushleft}
\end{flushleft}
\end{abstract}

\maketitle

\section{Introduction and main results}\label{section_intro}
A useful tool in the study of special functions related to quantum models of Calogero-Moser-Sutherland (CMS) type \cite{Cal71,Sut72,OP77} are the so-called \emph{kernel functions}. 
For the elliptic CMS type models \cite{OP77,LT12}, the kernel functions are of particular use for constructing different representations of the corresponding special functions \cite{langmann:7,LT12,AL18} and for studying their properties, for example spectral symmetries \cite{ruijsenaars2009}. 
The kernel functions have also proven to be an invaluable tool for studying the relativistic generalizations of the elliptic CMS models due to Ruijsenaars \cite{Rui87} and van Diejen \cite{vDi94}. It is the latter model which we consider in this paper.
The $N$-variable van Diejen model \cite{vDi94} with 9 (real) coupling parameters $g_{0},\ldots, g_{7}, \lambda \in \R$, elliptic modulus $\tau$ ($\Im(\tau) \gtr 0$), and ``relativistic deformation'' parameter $\beta \gtr 0$, is formally defined by the analytic difference operator\footnote{Throughout the paper, we write $A_{N}(\bs x ; \bs g , \lambda ,\beta)$, etc., to indicate the arguments $\bs x$ and the coupling parameters.}
\begin{equation}\label{eq_van_Diejen_operator}
A_{N}(\bs x ; \bs g , \lambda , \beta ) = \vartheta_{1}(\imag \lambda \beta) \sum_{\varepsilon=\pm} \sum_{j=1}^{N} V^{\varepsilon}_{j}(\bs x ; \bs{g}, \lambda , \beta) \exp\Bigl( - \varepsilon\imag \beta \frac{\partial}{\partial x_{j}} \Bigr) + V^{0}(\bs x ; \bs g , \lambda , \beta),
\end{equation}
where the coefficients are given by\footnote{We use the standard definition for the odd Jacobi theta function $\vartheta_{1}(x) := \vartheta_{1}(x \lvert \tau) 
$ ($\Im(\tau)\gtr 0$) found in \cite{WhitWat}; see also \eqref{eq_Jacobi_theta_function}.}
$$
V^{\pm}_{j} (\bs x ; \bs g , \lambda, \beta ) = \frac{\prod_{\nu=0}^{7} \vartheta_{1}( \pm x_{j} - \imag g_{\nu} \beta) }{\vartheta_{1}( \pm 2x_{j} ) \vartheta_{1}( \pm 2x_{j} - \imag \beta )} \prod_{\delta=\pm}  \prod_{ j^{\prime} \neq j}^{N} \frac{\vartheta_{1}( x_{j} + \delta x_{j^{\prime}} \mp \imag \lambda \beta )}{\vartheta_{1}(x_{j} + \delta x_{j^{\prime}} )} \quad (j=1,\ldots,N)
$$
and $V^{0}(\bs x ; \bs g , \lambda , \beta) $ an even elliptic function with simple poles at $x_{j} = \pm \imag \beta /2$ ($j=1,\ldots,N$) and congruent points w.r.t. half-period shifts. We have chosen to not write down the $V^{0}$ coefficient here and the precise definition is given in \eqref{eq_van_Diejen_potential}. Note that our normalization is different from the standard definition. (Here, and in the following, we often suppress the dependence of functions on the elliptic modulus $\tau$.) It is also known that van Diejen's model admits a commuting family of higher order difference operators \cite{KH97}. The operator \eqref{eq_van_Diejen_operator} generalizes Koornwinder's multivariate extension of the Askey-Wilson difference operator \cite{Koo92} in the ``additive variables'' convention. Koornwinder's operator can also be obtained from \eqref{eq_van_Diejen_operator}, after suitable rescaling of both parameters and the operator, in the trigonometric limit $\Im(\tau)\to\infty$.
 Other limiting cases were also considered by van Diejen \cite{vDi94}.

For our purposes, it will be more convenient to work with the operator in a different form: The operator is known to be symmetric with respect to an explicitly known weight function and can be transformed to a symmetric form by a similarity transformations. The symmetric version of \eqref{eq_van_Diejen_operator} is then given by
\begin{equation}
H_{N}(\bs x ; \bs{g} , \lambda , \beta) = \vartheta_{1}(\imag \lambda \beta) \sum_{\varepsilon=\pm} \sum_{j=1}^{N} V^{\varepsilon}_{j}(\bs x ; \bs{g}, \lambda , \beta)^{1/2} \exp\Bigl( - \varepsilon\imag \beta \frac{\partial}{\partial x_{j}} \Bigr) V^{-\varepsilon}_{j}( \bs x ; \bs{g}, \lambda , \beta)^{1/2} + V^{0}(\bs x ; \bs g , \lambda , \beta).
\label{eq_vD_Hamiltonian}\end{equation}

For the class of CMS models \cite{OP77}, there exists a remarkable functional identity which can be regarded as the source of all kernel function identities: Such \emph{source identities} imply known (groundstate) eigenvalue equations and all kernel function identities as special cases \cite{Sen96,HL10,Lan10,LT12}. 
Remarkably, the source identities can also be used to obtain the same type of identities for a mathematically natural generalization of the CMS models due to Chalykh, Feigin, and Veselov and Sergeev \cite{CFV98,Ser02,SV05,SV09a}, commonly referred to as \emph{deformed CMS models}. In a previous paper, we constructed source identities for the relativistic generalization of the $A_{N-1}$ elliptic CMS model due to Ruijsenaars \cite{Rui87} and obtained a deformed generalizations thereof. One of the main results of this paper is the introduction of a Chalykh-Feigin-Sergeev-Veselov (CFSV) type deformation of the van Diejen operator \eqref{eq_vD_Hamiltonian}. This \emph{deformed van Diejen operator} is given by
\begin{equation}
\begin{split}
H_{N,\tN}(\bs x , \bs\tx ; \bs g , \lambda , \beta) =& \sum_{\ve =\pm} \vartheta_{1}(\imag \lambda \beta)\sum_{j=1}^{N}  V^{\ve}_{j}(\bs x , \bs\tx)^{1/2} \exp\Bigl({- \ve \imag \beta \frac{\partial}{\partial x_{j}}}\Bigr) V^{-\ve}_{j}(\bs x , \bs\tx)^{1/2} \\
- & \vartheta_{1}(\imag \beta) \sum_{k=1}^{\tN} \tilde{V}^{\ve}_{k}(\bs x , \bs\tx)^{1/2} \exp\Bigl({+ \ve \imag \lambda \beta \frac{\partial}{\partial \tx_{k}}}\Bigr) \tilde{V}^{-\ve}_{k}(\bs x , \bs\tx)^{1/2} + V^{0}(\bs x , \bs\tx) 
\label{eq_deformed_van_Diejen_operator}
\end{split}
\end{equation}
with coefficients
\begin{equation*}
\begin{split}
V_{j}^{\pm}(\bs x , \bs\tx) =& \frac{\prod_{\nu=0}^{7} \vartheta_{1}( \pm x_{j} - \imag g_{\nu} \beta)}{\vartheta_{1}(\pm 2 x_{j}) \vartheta_{1}(\pm 2 x_{j} - \imag \beta)}\prod_{\delta=\pm} \prod_{\substack{ j'=1 \\ j' \neq j}}^{N} \frac{\vartheta_{1}(x_{j} + \delta x_{j'} \mp \imag \lambda \beta)}{\vartheta_{1}(x_{j} +\delta x_{j'})} \prod_{k=1}^{\tN} \frac{\vartheta_{1}(x_{j} + \delta \tx_{k} \mp \half \imag ( \lambda - 1 ) \beta)}{\vartheta_{1}(x_{j} + \delta \tx_{k} \mp \half \imag ( \lambda +1) \beta)},
\\
\tilde{V}_{k}^{\pm}(\bs x , \bs\tx) =& \frac{\prod_{\nu=0}^{7} \vartheta_{1}( \pm \tx_{k} - \imag (g_{\nu} - \half (\lambda+1)) \beta)}{\vartheta_{1}(\pm 2 \tx_{k}) \vartheta_{1}(\pm 2 \tx_{k} + \imag \lambda \beta)}\prod_{\delta=\pm} \prod_{ j=1}^{N} \frac{\vartheta_{1}(\tx_{k} + \delta x_{j} \mp \half \imag(\lambda -1) \beta)}{\vartheta_{1}(\tx_{k} +\delta x_{j} \pm \half \imag (\lambda +1 ) \beta)} \prod_{\substack{k'=1 \\ k' \neq k}}^{\tN} \frac{\vartheta_{1}( \tx_{k} + \delta \tx_{k'} \pm \imag  \beta)}{\vartheta_{1}(\tx_{k} + \delta \tx_{k'})},
\end{split}
\end{equation*}
and
$V^{0}(\bs x , \bs\tx)$ given in\footnote{The function $s(x)$ in \eqref{eq_deformed_operator_potential} should then be replaced with the odd Jacobi theta function $\vartheta_{1}(x)$.} \eqref{eq_deformed_operator_potential}.
Limiting cases of this operator was found in the works of Sergeev and Veselov\cite{SV09a}, and Feigin and Silantyev \cite{FS14}.
Furthermore, we construct explicit kernel functions for the deformed van Diejen operator \eqref{eq_deformed_van_Diejen_operator} and all of its limiting cases.

In this paper, we insist on having uniform arguments for all the different cases of Koornwinder-van Diejen type operators and their corresponding kernel functions. To this end, we introduce the function
\begin{equation}
s(x) = \begin{cases}
x \quad &\text{(rational case I)} \\
(1/r) \sin( r x) \quad &\text{(trigonometric case II)} \\
(a /\pi) \sinh( \pi x / a ) \quad &\text{(hyperbolic case III)} \\
(1/r) \e^{ r a / 4} \vartheta_{1}(r x \lvert \imag r a / \pi ) \quad &\text{(elliptic case IV)}
\end{cases}
\label{eq_s_function}
\end{equation}
where $r \gtr 0$ and $a \gtr 0$. While the elliptic case is the most general, in the sense that the other cases can be obtained by suitable limits, it is convenient for us to consider the different cases separately. The main reason for this is that the results for the elliptic case holds only under certain restrictions on the model parameters, the so-called \emph{balancing condition} below, but hold for arbitrary parameters in the other cases. 
In order to have uniform arguments as we proceed to treat these different cases, we need to introduce some further notation:
For each case, there is an additive subgroup $\Omega \subset \C$ consisting of all zeroes of $s(x)$. More specifically, we have that 
$$
\Omega=\bigoplus_{\nu=0}^{\rho} \Z \omega_{\nu} , \quad \text{where} \quad \rho= \begin{cases}0 \quad &\text{(I)} \\
1\quad &\text{(II) and (III)} \\
3 \quad  &\text{(IV)}
\end{cases},
$$
with 
$$
\omega_{0} = 0,\quad  \omega_{1} = \begin{cases} \pi / r \quad &\text{(II) and (IV)}\\
\imag a \quad &\text{(III)}
\end{cases}, \quad \omega_{2} = \imag a , \quad \omega_{3} = - \omega_{1} - \omega_{2}.
$$
The quasi-periodicity of $s(x)$ can also be expressed as \cite{KNS09}
\begin{equation}
s(x + \omega_{\nu}) = \epsilon_{\nu} \e^{ 2 \imag r \xi_{\nu} ( x + \half \omega_{\nu})} s(x)
\label{eq_quasi_periodicity}
\end{equation}
where $(\epsilon_{0},\epsilon_{1},\epsilon_{2},\epsilon_{3})=(1,-1,-1,-1)$ and $(\xi_{0},\xi_{1},\xi_{2},\xi_{3})=(0,0,-1,1)$ \cite{WhitWat}. 
Note that the quasi-periodicity factors are only non-trivial, in the sense that they are not $\pm1$, for the elliptic (IV) case. 

With our notation in place, we can define the following generalization of the Koornwinder-van Diejen type difference operator:
\begin{equation}
\cH(\bs X ; \bs m)= \sum_{\ve=\pm} \sum_{J=1}^{\cN} s(\imag \lambda m_{J} \beta) \cV^{\ve}_{J}(\bs X ; \bs m)^{1/2} \exp\Bigl( - \ve\imag  \frac{\beta}{m_{J}} \frac{\partial}{\partial X_{J}}\Bigr) \cV^{-\ve}_{J}(\bs X ; \bs m)^{1/2} + \cV^{0}(\bs X ; \bs m)
\label{eq_Sen_operator}
\end{equation}
where
\begin{multline}\label{eq_Sen_coeff_potential}
\cV^{0}(\bs X ; \bs m) = - \frac{1}{4} \Bigl( \prod_{\nu=1}^{\rho} s(\half \omega_{\nu}) \Bigr)^{2} \sum_{\nu=0}^{\rho} \frac{\exp(- r \xi_{\nu}( 2 \lambda \sum_{J=1}^{\cN} m_{J} + \sum_{\nu=0}^{2\rho +1} g_{\nu} - \half ( \rho + 1) ( \lambda + 1) ) \beta) }{\prod_{\mu \neq \nu}^{\rho} s(\half( \omega_{\nu} - \omega_{\mu}))} \\
\cdot \left( \frac{\prod_{\mu=0}^{2\rho+1} s( \half \omega_{\nu}+ \half \imag\beta - \imag g_{\mu} \beta)}{\prod_{\mu=0}^{\rho}
s(\half( \omega_{\nu} - \omega_{\mu} + \imag(1 - \lambda ) \beta))}\prod_{\delta=\pm} \prod_{J=1}^{\cN} \frac{s( \delta X_{J} + \half \omega_{\nu} + \quarter \imag [ \lambda(m_{J}-1) + 1/m_{J} + 1] \beta - \imag \lambda m_{J} \beta)}{s( \delta X_{J} + \half \omega_{\nu} + \quarter \imag [  \lambda(m_{J}-1) + 1/m_{J} + 1 ] \beta )}
\right. \\
\left. +
\frac{\prod_{\mu=0}^{2\rho+1} s( \half \omega_{\nu} + \half \imag \lambda \beta - \imag g_{\mu} \beta)}{
\prod_{\mu=0}^{\rho} s(\half( \omega_{\nu} - \omega_{\mu} + \imag ( \lambda - 1) \beta))} \prod_{\delta= \pm} \prod_{K=1}^{\cN}  \frac{s(\delta X_{K} + \half \omega_{\nu} + \quarter \imag [ \lambda(m_{K}+1) + 1/m_{K} - 1]\beta -\imag \lambda m_{K} \beta)}{s(\delta X_{K} + \half \omega_{\nu} + \quarter \imag [ \lambda(m_{K}+1) + 1/m_{K} - 1] \beta)} \right),
\end{multline}
\begin{equation}
\cV^{\pm}_{J}(\bs X ; \bs m) = \frac{\prod_{\nu=0}^{2\rho+1} s( \pm X_{J} - \imag d_{\nu,J} \beta)}{
s( \pm 2 X_{J}) s( \pm 2 X_{J} - \imag \beta / m_{J} )} \prod_{K \neq J}^{\cN} \prod_{\delta=\pm} f_{\pm}(X_{J} + \delta X_{K} ; m_{J}, m_{K})
\label{eq_Sen_coeff_shift}
\end{equation}
with
\begin{equation}
d_{\nu,J} = d( g_{\nu}, m_{J}) , \quad d(g,m) = \begin{cases} g \quad &\text{if } m=1 \text{ or } m = + 1/ \lambda \\
g - (\lambda + 1 ) / 2 \quad &\text{if } m=-1 \text{ or } m=-1/\lambda
\end{cases}
\label{eq_Sen_parameters}
\end{equation}
and
\begin{equation}\label{eq_Sen_coeff_pair}
f_{\pm}( x ; m , m') = \frac{s(x \mp \imag (m - m') ( \lambda m m' - 1 ) \beta /4 m m' \mp \imag \lambda m' \beta)}{s(x \mp \imag (m - m') ( \lambda m m' - 1 ) \beta /4 m m') },
\end{equation}
depending on the ``mass'' parameters $\bs m \in \Lambda^{\cN}$ taking values in the set
$$
\Lambda = \Bigl\{ 1 , -1 , -\frac{1}{\lambda} , + \frac{1}{\lambda} \Bigr\}.
$$

Throughout this paper, we will always assume that $\imag \beta / m$ is not equal to any periods $\omega_{\nu}$ $(\nu=0,\ldots,\rho)$ multiplied by a rational number, that is $(\imag \beta / m) \Z \cap \Omega = \emptyset$, for all $\beta \gtr 0$ and $m \in \Lambda$.

The pertinent eigenfunction for the operator in \eqref{eq_Sen_operator} is then given by
\begin{equation}
\Phi(\bs X ; \bs m) = \prod_{J=1}^{\cN} \psi(X_{J};m_{J}) \prod_{1 \leq J \less K \leq \cN} \prod_{\ve,\ve' = \pm} \phi( \ve X_{J} + \ve' X_{K} ; m_{J},m_{K})
\label{eq_Sen_function}
\end{equation}
where
\begin{equation}\label{eq_Sen_function_single}
\psi(x;m) = 
\left( \frac{G(2 x + \imag \beta/2m ; \beta/m) G(- 2 x + \imag \beta/2m ; \beta/m)}{\prod_{\nu=0}^{2\rho+1}G(x + \imag \beta/2 m - \imag d(g_{\nu},m) \beta; \beta/m)G(-x + \imag \beta/2 m - \imag d(g_{\nu},m) \beta; \beta/m)}\right)^{1/2},
\end{equation}
and
\begin{equation}\label{eq_Sen_function_pair}
\phi(x ; m , m') = \begin{dcases} 
\left( \frac{G(x + \imag \beta/2 m ; \beta / m)}{G(x - \imag \lambda m \beta + \imag \beta/2m; \beta/m)}\right)^{1/2}\quad &\text{if } m'=m \\
G(x - \imag \lambda m \beta / 2; \beta / m) \quad &\text{if } m'=- m\\
s(x)^{1/2} \quad &\text{if } m'=+1/\lambda m\\
s(x - \imag\lambda m \beta /2 + \imag \beta /2 m)^{-1/2} \quad &\text{if } m'=-1/\lambda m
\end{dcases}
\end{equation}
with $G(x;\alpha)$ a function of two variables $x\in\C$ and $\alpha \in \C$ such that $\Re(\alpha)\neq 0$, satisfying the functional equation
\begin{equation}
\frac{G(x + \imag \alpha /2 ; \alpha)}{G(x - \imag \alpha/2; \alpha)} = c \cdot s(x) \quad (c \in \C^{\ast}:=\C\setminus\{0\}).
\label{eq_Gamma_function_equation}
\end{equation}
It is known that functions satisfying the functional equation exists in all cases (I)--(IV)  \cite{WhitWat,Rui97}. For example in the rational case (I), this function can be expressed in terms of the Euler $\Gamma$-function and we refer to the function $G(x;\alpha)$ as the \emph{Gamma function} in all cases. A more detailed, albeit not complete, discussion on the Gamma function is given in Section~\ref{section_gamma_function}.

Before presenting our main result, let us introduce the notation
$$
\abs{\bs m } = \sum_{J=1}^{\cN} m_{J} , \quad \abs{\bs g} = \sum_{\nu=0}^{2 \rho +1 } g_{\nu}, \quad \abs{\bs \omega} = \sum_{\nu=0}^{\rho} \omega_{\nu}
$$
in order to simplify the formulas below. It is worth noting that the sum $\abs{\bs\omega}$ is only non-zero in the trigonometric (II) and hyperbolic (III) cases.

\begin{theorem}\label{thm_source_identity}
Let $\cN\in\Z_{\geq0}$, $\bs X=(X_{1},\ldots,X_{\cN})\in (\C\setminus\Omega)^{\cN}$ with complex variables $X_{J}$ such that $X_{J}\neq \pm X_{K} (\text{\emph{mod} } \Omega)$ $\forall K\neq J$, and $\bs m \in \Lambda^{\cN}$. Then
\begin{equation}
\Bigl( \cH_{\cN}(\bs X ; \bs m) - \bigl[  \frac{\prod_{\nu=1}^{\rho}s(\half \omega_{\nu})}{2} \bigr]^{2} s(  \imag \beta [ 2 \lambda \abs{\bs m} + \abs{\bs g} - \frac{1}{2} (\rho+1)(\lambda+1)] - \abs{\bs\omega}) \Bigr) \Phi(\bs X ; \bs m) = 0
\label{eq_Source_identity}
\end{equation}
holds true in the rational \emph{(I)}, trigonometric \emph{(II)}, and hyperbolic \emph{(III)} cases, while in the elliptic \emph{(IV)} case it holds if and only if the balancing condition 
\begin{equation}
 2 \lambda \sum_{J=1}^{\cN} m_{J} +  \sum_{\nu=0}^{7} g_{\nu} - 2 (\lambda+1) =0 \quad \text{\emph{(IV)}}
\label{eq_balancing_condition}
\end{equation}
is satisfied.
\end{theorem}
\noindent(Proof of Theorem \ref{thm_source_identity} is given in Section~\ref{section_proof_of_theorem}.)

It is also straightforward to deduce the following result.
\begin{lemma}\label{lemma_Sen_conjugate_operator}
Let $\cN \in \Z_{\geq0}$, $\bs X \in (\C\setminus\Omega)^{\cN}$ with complex variables $X_{J}$ such that $X_{J}\neq X_{K} (\text{mod } \Omega)$ $\forall K \neq J$ and $\bs m \in \Lambda^{\cN}$. For $\cH_{\cN}(\bs X ; \bs m)$ in \eqref{eq_Sen_operator} and $\Phi( \bs X ; \bs m)$ in \eqref{eq_Sen_function},\eqref{eq_Sen_function_single}, and \eqref{eq_Sen_function_pair}, the analytic difference operator
\begin{equation}\label{eq_Sen_operator_reduced}
\cA_{\cN}(\bs X ; \bs m) = \Phi(\bs X; \bs m)^{-1} \circ \cH_{\cN}(\bs X ; \bs m) \circ \Phi(\bs X; \bs m)
\end{equation}
is identical with 
\begin{equation}\label{eq_lemma_1}
\sum_{\ve = \pm} \sum_{J=1}^{\cN} s(\imag \lambda m_{J} \beta) \cV^{\ve}_{J}(\bs X ; \bs m) \exp\Bigl( - \ve \imag \frac{\beta}{m_{J}} \frac{\partial}{\partial X_{J}} \Bigr) + \cV^{0}(\bs X ; \bs m), 
\end{equation}
with $\cV_{J}^{\pm}$ in \eqref{eq_Sen_coeff_shift} and $\cV^{0}$ in \eqref{eq_Sen_coeff_potential}, for all cases \emph{(I)}-\emph{(IV)}. Furthermore, the operator $\cA_{\cN}(\bs X ; \bs m)$ also equals
$$
\sum_{\ve = \pm} \sum_{J=1}^{\cN} s(\imag \lambda m_{J} \beta) \cV_{J}^{\ve}(\bs X ; \bs m ) \Bigl( \exp\bigl( - \ve \imag \frac{\beta}{m_{J}} \frac{\partial}{\partial X_{J}} \bigr) -1 \Bigr)  + \frac{\prod_{\nu=1}^{\rho}s(\half \omega_{\nu})^{2}}{4} s( \imag \beta [ 2 \lambda \abs{\bs m} + \abs{\bs g} - \frac{1}{2} (\rho +1 ) ( \lambda+1)] - \abs{\bs\omega})
$$
in the rational \emph{(I)}, trigonometric \emph{(II)}, and hyperbolic \emph{(III)} cases, while in the elliptic \emph{(IV)} case only if the parameters satisfy the balancing condition in \eqref{eq_balancing_condition}.
\end{lemma}
\noindent(Proof of Lemma~\ref{lemma_Sen_conjugate_operator} is given in Section~\ref{section_proof_of_lemma}.)

The source identity implies various interesting identities as special cases, which we will discuss in Section~\ref{section_special_cases}. Among these identities, we find the known kernel function identities for the Koornwinder-van Diejen operators; see Corollaries \ref{cor_vD_Cauchy_KFI} and \ref{cor_vD_dual_Cauchy_KFI}. The most general of these identities is for the {deformed van Diejen operators} in \eqref{eq_deformed_van_Diejen_operator}: The operator $H_{N,\tN}(\bs x , \bs\tx ; \bs g , \lambda , \beta)$ \eqref{eq_deformed_van_Diejen_operator} is obtained in the special case where $\cN = N + \tN$ ($N,\tN\in\Z_{\gtr0}$) and 
\begin{equation}\label{eq_deformed_case}
(m_{J},X_{J})= \begin{cases}( 1 , x_{J} ) \quad &\text{for } J=1,\ldots,N\\
( -1/\lambda , \tx_{J-N} ) \quad &\text{for } J-N=1,\ldots,\tN
\end{cases}
\end{equation}
in \eqref{eq_Sen_operator}, while the most general case, given when $\cN = N + \tN + M + \tM$ ($N,\tN,M,\tM\in\Z_{\gtr0}$) and
\begin{equation}
(m_{J},X_{J})= \begin{cases}( 1 , x_{J} ) \quad &\text{for } J=1,\ldots,N\\
( -1/\lambda , \tx_{J-N} ) \quad &\text{for } J-N=1,\ldots,\tN\\
(-1 , y_{J-N-\tN}) \quad &\text{for } J- N - \tN = 1 ,\ldots, M\\
(+1/\lambda , \ty_{J-N-\tN-M})\quad &\text{for } J-N-\tN-M = 1,\ldots, \tM
\end{cases}
\label{eq_deformed_KF_case}
\end{equation}
in \eqref{eq_Source_identity}, yields a kernel function identity for a pair of deformed Koornwinder-van Diejen operators; see Corollary~\ref{cor_deformed_KFI}.

The plan of the paper is as follows: We start by recalling the Gamma functions used in this paper in Section~\ref{section_gamma_function}. The proof of Theorem~\ref{thm_source_identity} and Lemma~\ref{lemma_Sen_conjugate_operator} is given in Section~\ref{section_proof}. In Section~\ref{section_special_cases}, we state the various special cases of Theorem~\ref{thm_source_identity}, including kernel function identities. We conclude with some remarks in Section~\ref{section_final}.

\section{On the Gamma functions}\label{section_gamma_function}

In this Section, we collect the definitions needed for the Gamma functions used in this paper. For a more detailed discussion, we refer the reader to, for example, \cite{Rui97,KNS09}. The Gamma functions $G(x;\alpha)$ are defined by \eqref{eq_Gamma_function_equation}, but this functional equation \eqref{eq_Gamma_function_equation} does not have a unique solution. Let us illustrate this with a simple example: Suppose that $G(x;\alpha)$ satisfies \eqref{eq_Gamma_function_equation}, then the function $Q(\sinh( 2 \pi x / \alpha)) G(x;\alpha)$ also satisfies \eqref{eq_Gamma_function_equation} for any meromorphic function $Q(x)$ due to the periodicity of the $\sinh$-function. More generally, $G(x;\alpha)$ is only defined up to multiplication by any $\imag \alpha$-periodic function. We refer to such functions as \emph{quasi-constants}. The issue of how to fix the quasi-constant was addressed by Ruijsenaars in \cite{Rui97} with the construction of so-called minimal solutions. As such, we restrict our attention to particular solutions of \eqref{eq_Gamma_function_equation}. 
We wish to stress that although quasi-constants do not affect the difference equation \eqref{eq_Gamma_function_equation}, they can alter the analytical properties of the Gamma function, and by extension the kernel functions, to a large extent.

Furthermore, it is worth noting that if $G_{1}(x;\alpha)$ is a functions satisfying \eqref{eq_Gamma_function_equation}, then function $G_{2}(x;\alpha) = G_{1}(-x ; -\alpha)$ also satisfies \eqref{eq_Gamma_function_equation}, but for a different constant $c$, as can be checked using \eqref{eq_Gamma_function_equation}. These two solutions of \eqref{eq_Gamma_function_equation} will not necessarily only differ by a quasi-constant. Suppose then that there is a non-zero meromorphic function $G_{1}(x;\alpha)$ for $\Re(\alpha) \gtr 0$, then it is clear that the function $G_{2}$ is defined for $\Re(\alpha)\less 0$ and as such is a different extensions of the Gamma function. We refer the reader to \cite[Appendix~A]{AHL14} for a more detailed discussion. Following \cite{AHL14}, we define our Gamma function $G(x; \alpha)$ for $\Re(\alpha)\neq 0$ by 
\begin{equation}
G(x;\alpha) = \begin{cases}G_1(x;\alpha) \quad \text{for } \Re(\alpha)\gtr 0\\
G_{2}(x; \alpha)\quad \text{for } \Re(\alpha) \less 0,
\end{cases}
\label{eq_Gamma_function_continuation}
\end{equation}
and give the $G_{1}$ Gamma functions that satisfy \eqref{eq_Gamma_function_equation} for $\Re(\alpha)\gtr 0$ explicitly at the end of this section.
It is then clear that the Gamma function $G(x;\alpha)$ satisfies
\begin{equation}
G(x ; - \alpha) = G(- x ; \alpha) \quad (\Re(\alpha) \neq 0).
\label{eq_Gamma_minus_alpha}
\end{equation}

We then finish this section by giving\footnote{These Gamma functions are the same as those Gamma functions used in \cite{AHL14}.} the $G_1$ functions that satisfy \eqref{eq_Gamma_function_equation} for $\Re(\alpha)\gtr0$.

\subsection{Rational (I) case}
We recall that the Euler $\Gamma$-function, defined for example by the infinite product
$$
\Gamma(x) = \frac{1}{x} \prod_{n\in\Z_{\gtr0}} \frac{(1 - 1/n)^x}{1+ x/ n}
$$
for all complex $x \in \C\setminus\Z_{\leq 0}$ \cite{WhitWat}, satisfies the difference equation
$$
\Gamma(x + 1) = x \Gamma(x).
$$
In the rational case, we have that 
$$
G_{1}(x; \alpha) = \Gamma( 1/2 + x / \imag \alpha)
$$
satisfies \eqref{eq_Gamma_function_equation} with $c= 1/\imag \alpha$. In this case, it is clear that $G_{1}(x;-\alpha) = G_{1}(-x ; \alpha)$ holds without our restriction \eqref{eq_Gamma_minus_alpha}. 

\subsection{Trigonometric (II) case}
A standard choice for the trigonometric Gamma function is given by \cite{Rui97}
$$
G_{\text{R}}(r, \alpha ; x) = \prod_{n\in\Z_{\gtr 0}} ( 1 - \e^{- r \alpha ( 2 n - 1 )} \e^{2 \imag r x})^{-1} \quad (r , \Re(\alpha)\gtr 0 )
$$
which satisfies the functional equation
$$
\frac{G_{\text{R}}(r, \alpha ; x + \imag \alpha / 2)  }{G_{\text{R}}(r, \alpha ; x - \imag \alpha / 2) } = (1 - \e^{2 \imag r x}).
$$

We then define the $G_{1}$ function as
$$
G_{1}(x; \alpha) = \exp({- r x^{2} / 2 \alpha}) G_{\text{R}}( r , \alpha ; x ) \quad (r, \Re(\alpha)\gtr0)
$$
which satisfies \eqref{eq_Gamma_function_equation} for $c= -2 \imag r$.


\subsection{Hyperbolic (III) case}
A standard choice for the hyperbolic Gamma function is given by \cite{Rui97}
$$
G_{\text{R}}(a , \alpha ; x ) = \exp\Bigl( \int_{0}^{\infty} \frac{dy}{y} \bigl( \frac{\sin( 2 x y) }{2 \sinh( a y) \sinh( \alpha y)} - \frac{x}{a \alpha y} \bigr) \Bigr) \quad ( \abs{\Im( x )} \less \Re( a + \alpha ) / 2, \Re(a), \Re(\alpha) \gtr 0),
$$
which satisfies the functional equation
$$
\frac{G_{\text{R}}(a , \alpha ; x + \imag \alpha / 2) }{G_{\text{R}}(a , \alpha ; x - \imag \alpha / 2) } = 2 \cosh( \pi x / a).
$$
We then define 
$$
G_{1}(x ; \alpha) = G_{\text{R}}(a , \alpha ; x - \imag a /2 ) \quad ( \abs{ \Im(x)- \imag a /2  } \less \Re( a + \alpha ) /2 )
$$
which satisfies \eqref{eq_Gamma_function_equation} with $c = - 2 \pi \imag / a$.

\subsection{Elliptic (IV) case}
A standard choice for the elliptic Gamma function is given by \cite{Rui97}
$$
G_{\text{R}}(r , a , \alpha ; x) = \prod_{n,m \in \Z_{\gtr0}} \frac{(1- \e^{-r a ( 2n - 1 )} \e^{- r \alpha ( 2 m -1 )} \e^{-2 \imag r x})}{(1- \e^{-r a ( 2n - 1 )} \e^{- r \alpha ( 2 m -1 )} \e^{2 \imag r x})} \quad (r , \Re(a),\Re(\alpha) \gtr 0)
$$
which satisfies the functional equation 
$$
\frac{G_{\text{R}}(r , a , \alpha ; x+ \imag \alpha / 2)}{G_{\text{R}}(r , a , \alpha ; x - \imag \alpha / 2)} = \prod_{n\in\Z_{\gtr0}} (1 - \e^{- r a ( 2n -1 )} \e^{-2 \imag r x})(1- \e^{-r a ( 2 n - 1) } \e^{2 \imag r x}) \quad (r , \Re(a ) \gtr 0).
$$
Recalling that the (odd) Jacobi theta function has the product representation \cite{WhitWat}
\begin{equation}
\vartheta_{1}(r x \lvert \tau ) = 2 \e^{\quarter \imag \pi \tau} \sin( r x ) \prod_{n\in\Z_{\gtr 0}}(1 - \e^{ 2 \imag \pi n \tau }) (1- \e^{2 \imag \pi n \tau } \cos( 2 r x) + \e^{4 \imag \pi n \tau}) \quad (\Im(\tau) \gtr 0),
\label{eq_Jacobi_theta_function}
\end{equation}
we then define 
$$
G_{1}(x ; \alpha) = \exp( - r x^{2} / 2 \alpha ) G_{\text{R}}(r , a , \alpha ; x - \imag a / 2)
$$
which satisfies \eqref{eq_Gamma_function_equation} for 
$$
c = - \imag \e^{\quarter r a} \prod_{n\in\Z_{\gtr 0}} (1 - \e^{- 2 r n a})^{-1}.
$$

\section{Proof of main results}\label{section_proof}

Before giving the proof of the source identity, and subsequent results, we need the following key functional identity.
\begin{lemma}\label{lemma_key_lemma}
Let $\cN\in\Z_{\geq0}$, $\bs X, \bs m, \bs a \in\C^{\cN}$, $\gamma\in\C$ such that $\Im(\gamma)\neq 0$, $\bs c=(c_{0},\ldots, c_{\rho}), \bs d=(d_{0},\ldots,d_{\rho}) \in \C^{\rho+1}$ be $(\text{\emph{mod} } \Omega)$ distinct complex numbers, and $\bs n =(n_{0},\ldots,n_{2 \rho +1}) \in \C^{2 \rho + 2}$ with $s(x)$, $\rho$, and $\bs\omega=(\omega_{0},\ldots,\omega_{\rho})$ as defined above. Then 
\begin{equation}
\begin{split}
\sum_{\ve = \pm} & \sum_{J=1}^{\cN} s( \gamma m_{J} ) \prod_{\delta = \pm} \prod_{K\neq J}^{\cN} \frac{s( X_{J} + \delta X_{K} + \ve( a_{J} - a_{K} - \gamma m_{K}))}{s(X_{J} + \delta X_{K} + \ve( a_{J} - a_{K}) )}\\ 
&\cdot \prod_{\nu=0}^{\rho} \frac{s(\ve X_{J} - \half \gamma m_{J} - \half \omega_{\nu}) s( \ve X_{J} + a_{J} - c_{\nu} -\half \omega_{\nu} - n_{\nu}) s(\ve X_{J} + a_{J} - d_{\nu} - \half \omega_{\nu} - n_{\nu+\rho+1})}{ s(\ve X_{J} - \half \omega_{\nu}) s(\ve X_{J} + a_{J} - c_{\nu} - \half \omega_{\nu}) s(\ve X_{J} + a_{J} - d_{\nu} - \half \omega_{\nu}) } \\
- \sum_{\nu=0}^{\rho} &\Bigl\{\bigl( \prod_{\delta = \pm} \prod_{J=1}^{\cN} \frac{s( \delta X_{J} + \half \omega_{\nu} +c_{\nu} - a_{J} - \gamma m_{J})}{s(\delta X_{J} + \half \omega_{\nu} +c_{\nu} - a_{J})} \bigr) \frac{\prod_{\mu=0}^{\rho}s( \half( \omega_{\nu} - \omega_{\mu}) + c_{\nu} - c_{\mu} -  n_{\mu})}{\prod_{\mu\neq\nu}^{\rho} s( \half( \omega_{\nu} - \omega_{\mu}) + c_{\nu} - c_{\mu})} \Bigr. \\
&\cdot \bigl( \prod_{\mu=0}^{\rho} \frac{s(\half( \omega_{\nu}-\omega_{\mu}) + c_{\nu} - d_{\mu} - n_{\mu+\rho +1})}{s(\half( \omega_{\nu}-\omega_{\mu}) + c_{\nu} - d_{\mu})}\bigr) + \bigl( \prod_{\delta = \pm} \prod_{J=1}^{\cN} \frac{s( \delta X_{J} +\half \omega_{\nu} + d_{\nu}- a_{J} - \gamma m_{J})}{s( \delta X_{J} +\half \omega_{\nu} + d_{\nu}- a_{J})} \bigr) \\ \Bigl.
&\cdot \bigl(\prod_{\mu=0}^{\rho} \frac{s(\half(\omega_{\nu}-\omega_{\mu}) + d_{\nu} - c_{\mu} - n_{\mu})}{s(\half(\omega_{\nu}-\omega_{\mu}) + d_{\nu} - c_{\mu})}\bigr) \frac{\prod_{\mu=0}^{\rho}s(\half( \omega_{\nu} - \omega_{\mu}) + d_{\nu} - d_{\mu} - n_{\mu+\rho +1})}{\prod_{\mu\neq\nu}^{\rho} s(\half( \omega_{\nu} - \omega_{\mu}) + d_{\nu} - d_{\mu})} \Bigr\} \\
= & s\Bigl( 2 \gamma \sum_{J=1}^{\cN} m_{J} + \sum_{\nu=0}^{2\rho+1} n_{\nu} \Bigr)
\label{eq_Key_eq}
\end{split}
\end{equation}
holds as an identity for a meromorphic function in $\bs X$ in the rational \emph{(I)}, trigonometric \emph{(II)}, and hyperbolic \emph{(III)} cases. In the elliptic \emph{(IV)} case, \eqref{eq_Key_eq} holds true if and only if the relation
\begin{equation}\label{eq_lemma_balancing_condition}
2 \gamma \sum_{J=1}^{\cN}m_{J} + \sum_{\nu=0}^{7} n_{\nu} =0 \quad \text{\emph{(IV)}}
\end{equation}
is satisfied.
\end{lemma}

It is not clear to us whether the result in Lemma~\ref{lemma_key_lemma} can be found in the literature. A similar identity was used in \cite{Rui09} and \cite{KNS09} to prove the kernel function identity for the the van Diejen operator \eqref{eq_van_Diejen_operator}. It is also possible that \eqref{eq_Key_eq} can be obtained from the determinant identities of type $BC$ \cite{Mas13,KMN16}. As this is a key Lemma for our main results, we proceed to give a short proof using Liouville's Theorem.

\begin{proof}[Proof of Lemma~\ref{lemma_key_lemma}]
Consider the l.h.s. of \eqref{eq_Key_eq} as a function of $X_{1}$ (say) by fixing all the other variables to generic values. 
It is then clear that this function has simple poles determined by the zeroes of the $s$-functions in the denominators. Computing the residues at these points, we find that they all cancel in the rational (I), trigonometric (II), and hyperbolic (III) cases and that the l.h.s. of \eqref{eq_Key_eq} is an entire function of $X_{1}$. We also find that the function is bounded as $X_{1}\to \infty$ (I), $X_{1}\to - \imag \infty$ (II), and $X_{1}\to \infty$ (III) respectively. By Liouville's Theorem, we have that this function is a constant in $X_{1}$ (say) in the rational (I), trigonometric (II), and hyperbolic (III) cases. Since the l.h.s. of \eqref{eq_Key_eq} is invariant under any permutations of the $X_{J}$-variables, we have that the l.h.s. of \eqref{eq_Key_eq} is a constant in all these cases. Taking suitable limits of the variables (same as the $X_{1}$ limits above) yields the r.h.s. of \eqref{eq_Key_eq} for the rational (I), trigonometric (II), and hyperbolic (III) cases.

For the elliptic (IV) case, let us express the l.h.s. of \eqref{eq_Key_eq} as
$$
\sum_{J=1}^{\cN} s( \gamma m_{J} ) (C^{+}_{J}(X_{1}) + C^{-}_{J}(X_{1})) - \sum_{\nu=0}^{3}C^{0}_{\nu}(X_{1}) ,
$$
where we suppress the dependence on the other variables and parameters.
By straightforward calculations, using \eqref{eq_quasi_periodicity}, we find that $C^{0}_{\nu}$ and $C^{\pm}_{J}$, for $J=2,\ldots,\cN$, are elliptic functions of $X_{1}$ while
$$
C_{1}^{\pm}(X_{1} + \omega_{\mu}) = \exp\bigl(\mp 2 \imag r \xi_{\mu}( 2 \gamma \sum_{J=1}^{\cN} m_{J} + \sum_{\nu=0}^{7} n_{\nu})\bigr) C_{1}^{\pm}(X_{1}) \quad(\mu=0,\ldots,3).
$$
Imposing the balancing condition \eqref{eq_lemma_balancing_condition}, we find by straightforward calculations that all the residues vanish and that the l.h.s. of \eqref{eq_Key_eq} is therefore an elliptic function without any poles and thus a constant. Taking the limit $X_{1}\to0$, we find that this constant is zero, which equals the r.h.s. of \eqref{eq_Key_eq} under \eqref{eq_lemma_balancing_condition}.
\end{proof}

\subsection{Proof of Theorem \ref{thm_source_identity}}\label{section_proof_of_theorem}
We are now in a position to prove the main Theorem.
\begin{proof}[Proof of Theorem~\ref{thm_source_identity}] 
Assuming, for now, that there exists a non-zero meromorphic function $\Phi(\bs X; \bs m)$ that satisfies an eigenvalue equation for the operator in \eqref{eq_Sen_operator}, we make the ansatz in \eqref{eq_Sen_function} and compute
\begin{multline}\label{eq_proof_eigenvalue_identity}
\cH_{\cN}(\bs X ; \bs m) \Phi(\bs X ; \bs m) = \Bigl[ \cV^{0}(\bs X ; \bs m) +  \sum_{\ve = \pm} \sum_{J=1}^{\cN} s(\imag \lambda m_{J} \beta) \Bigr. \\ 
\cdot \Bigl( \frac{\prod_{\nu=0}^{2\rho +1 } s( \ve X_{J} - \imag d_{\nu,J} \beta ) s(-\ve X_{J} + \imag \beta /m_{J} - \imag d_{\nu,J}\beta)}{s(2\ve X_{J}) s(2 \ve X_{J} - \imag \beta/ m_{J}) s(-2 \ve X_{J} + 2 \imag \beta/m_{J}) s(-2\ve X_{J} + \imag \beta / m_{J})} \Bigr)^{\half} \frac{\psi( X_{J} - \imag \ve \beta/m_{J} ; m_{J})}{\psi( X_{J} ; m_{J})}
\\
\cdot\prod_{\delta=\pm}\prod_{K \less J}\Bigl\{\Bigl(  f_{\ve}(X_{J} + \delta X_{K} ; m_{J} , m_{K}) f_{-\ve}(X_{J} + \delta X_{K} - \imag \ve \beta/m_{J} ; m_{J},m_{K})\Bigr)^{\half} \Bigr. \\
\Bigl. \cdot \frac{\phi( X_{J} + \delta X_{K} -  \ve \imag \beta/m_{J}; m_{J} , m_{K}) \phi( - X_{J} + \delta X_{K} +  \ve \imag \beta/m_{J}; m_{J} , m_{K}) }{\phi( X_{J} + \delta X_{K}; m_{J} , m_{K})\phi( -X_{J} + \delta X_{K}; m_{J} , m_{K})  } \Bigr\}
\\
\cdot\prod_{\delta=\pm}\prod_{K \gtr J} \Bigl\{\Bigl(  f_{\ve}(X_{J} + \delta X_{K} ; m_{J} , m_{K}) f_{-\ve}(X_{J} + \delta X_{K} - \imag \ve \beta/m_{J} ; m_{J},m_{K})\Bigr)^{\half} \Bigr. \\ \Bigl. \Bigl.
\cdot \frac{\phi( X_{J} + \delta X_{K} -  \ve \imag \beta/m_{J}; m_{K} , m_{J}) \phi( - X_{J} + \delta X_{K} +  \ve \imag \beta/m_{J}; m_{K} , m_{J}) }{\phi( X_{J} + \delta X_{K}; m_{K} , m_{J})\phi( -X_{J} + \delta X_{K}; m_{K} , m_{J}) }\Bigr\}\Bigr] \Phi(\bs X ; \bs m).
\end{multline}
We then determine the functions $\psi(x;m)$ and $\phi(x ; m , m')$ $(m,m'\in\Lambda)$ such that we can apply the results of Lemma~\ref{lemma_key_lemma}. Let us start by determining the functions $\phi(x;m,m')$ such that they satisfy the relations 
\begin{subequations}
\begin{multline}\label{eq_proof_relation_1}
\prod_{\delta=\pm} f_{\pm}(X_{J} + \delta X_{K} ; m_{J} , m_{K})^{\half} f_{\mp}(X_{J} + \delta X_{K} \mp \imag \beta/m_{J} ; m_{J},m_{K})^{\half} \\  \cdot
\frac{\phi( X_{J} + \delta X_{K} \mp \beta/m_{J}; m_{J} , m_{K}) \phi( - X_{J} + \delta X_{K} \pm \beta/m_{J}; m_{J} , m_{K}) }{\phi( X_{J} + \delta X_{K}; m_{J} , m_{K})\phi( -X_{J} + \delta X_{K}; m_{J} , m_{K})  } \\
=  \prod_{\delta = \pm} \frac{s( X_{J} + \delta X_{K} + \ve( a_{J} - a_{K} - \gamma m_{K}))}{s(X_{J} + \delta X_{K} + \ve( a_{J} - a_{K}) )} \quad ( \ve = \pm )
\end{multline}
for $K \less J$ and 
\begin{multline}\label{eq_proof_relation_2}
\prod_{\delta=\pm} f_{\pm}(X_{J} + \delta X_{K} ; m_{J} , m_{K})^{\half} f_{\mp}(X_{J} + \delta X_{K} \mp \imag \beta/m_{J} ; m_{J},m_{K})^{\half} \\  \cdot
\frac{\phi( X_{J} + \delta X_{K} \mp \beta/m_{J}; m_{K} , m_{J}) \phi( - X_{J} + \delta X_{K} \pm \beta/m_{J}; m_{K} , m_{J}) }{\phi( X_{J} + \delta X_{K}; m_{K} , m_{J})\phi( -X_{J} + \delta X_{K}; m_{K} , m_{J})  } \\
=  \prod_{\delta = \pm} \frac{s( X_{J} + \delta X_{K} + \ve( a_{J} - a_{K} - \gamma m_{K}))}{s(X_{J} + \delta X_{K} + \ve( a_{J} - a_{K}) )} \quad ( \ve = \pm )
\end{multline}
\end{subequations}
for $K \gtr J$, for suitable constants $\gamma$ and $\bs a$ to be determined. Also, assume that $a_{J} = a(m_{J})$ ($J=1,\ldots, \cN$) for some function $a(m)$. We proceed by considering the conditions \eqref{eq_proof_relation_1} and \eqref{eq_proof_relation_2} in the different cases where $(m_{J},m_{K})=(m,m)$, $(m,-m)$, $(m,+1/\lambda m)$, and $(m,-1/\lambda m)$ for all $m\in \Lambda$.

When $(m_{J},m_{K})=(m,m)$, then both \eqref{eq_proof_relation_1} and \eqref{eq_proof_relation_2} become
\begin{multline}
\frac{\phi( X_{J} + \delta X_{K} \mp \imag \beta/2m; m , m) \phi( - X_{J} + \delta X_{K} \pm \beta/2m; m , m) }{\phi( X_{J} + \delta X_{K} \pm \imag \beta / 2 m ; m , m)\phi( -X_{J} + \delta X_{K}\mp \imag \beta / 2 m; m , m)  } = \\ 
 \prod_{\delta=\pm} \frac{s(X_{J} + \delta X_{K} \pm \imag \beta / 2 m - \ve\gamma m)}{s(X_{J} + \delta X_{K} \pm \imag \beta / 2 m )} \Bigl( \frac{s( X_{J} + \delta X_{K} \pm \imag \beta / 2 m) s( X_{J} + \delta X_{K} \mp \imag \beta /2 m)}{s(X_{J} + \delta X_{K} \pm \imag \beta / 2 m \mp \imag \lambda m \beta)s(X_{J} + \delta X_{K} \mp \imag \beta /2 m  \pm \imag \lambda m \beta)} \Bigr)^{\half},
\end{multline}
by shifting the $X_{J}$ variable by $\pm \imag \beta /2m$. The choice
\begin{equation}\label{eq_proof_constants}
\ve = \pm ,\quad \gamma = \imag \lambda \beta
\end{equation}
reduces the equations to
\begin{multline}
\prod_{\delta=\pm}\frac{\phi( X_{J} + \delta X_{K} \mp \imag \beta/2m; m , m) \phi( - X_{J} + \delta X_{K} \pm \imag \beta/2m; m , m) }{\phi( X_{J} + \delta X_{K} \pm \imag \beta / 2 m ; m , m)\phi( -X_{J} + \delta X_{K}\mp \imag \beta / 2 m; m , m)  } \\
= \prod_{\delta=\pm}\Bigl( \frac{s( X_{J} + \delta X_{K} \mp \imag \beta /2 m) s(X_{J} + \delta X_{K} \pm \imag \beta / 2 m \mp \imag \lambda m \beta ) }{s(X_{J} + \delta X_{K} \pm \imag \beta / 2 m )s(X_{J} + \delta X_{K} \mp \imag \beta /2 m  \pm \imag \lambda m \beta)} \Bigr)^{\half}
\end{multline}
which have the common solution $\phi(x ; m , m)$ in \eqref{eq_Sen_function_pair}.

When $(m_{J},m_{K} )= (m,-m)$ for $m\in\Lambda$, then \eqref{eq_proof_relation_1} and \eqref{eq_proof_relation_2}, using \eqref{eq_proof_constants}, become
\begin{multline*}
\prod_{\delta= \pm} \frac{\phi(X_{J} + \delta X_{K} \mp \imag \beta / 2 m ; m , -m)\phi(-X_{J} + \delta X_{K} \pm \imag \beta /2 m ; m , -m)}{\phi(X_{J} + \delta X_{K}  \pm\imag \beta /2 m ; m , -m)\phi(-X_{J} + \delta X_{K}  \mp \imag\beta /2 m; m , -m)} \\
= \prod_{\delta=\pm} \frac{s(X_{J} + \delta X_{K} \pm (a(m) - a(-m))  \pm \beta /2 m \pm \imag \lambda m \beta)}{s(X_{J} + \delta X_{K} \pm (a(m) - a(-m))  \pm \beta /2 m)}
\end{multline*}
by shifting the $X_{J}$ variable by $ \pm \imag \beta /2 m$. (Note that the functions $f_{\pm}$ in \eqref{eq_Sen_operator} satisfy \newline $f_{\pm}(x ; m , m' ) f_{\mp}(x \mp \imag \beta/m  ; m ,m' ) = 1$ when $m' = -m$ or $m'=+1/\lambda m$ for all $m\in\Lambda$.) 
Choosing 
\begin{equation}
\label{eq_a_cond_1}
a(m)-a(-m) = -\imag \beta / 2 m - \imag \lambda m \beta /2  \quad (\forall m\in\Lambda),
\end{equation} yields that $\phi(x; m ,-m)$ should satisfy the difference equations
\begin{equation}
\prod_{\delta=\pm}\frac{\phi( X_{J} + \delta X_{K} \mp \imag \beta/2m; m , -m) \phi( - X_{J} + \delta X_{K} \pm \imag  \beta/2m; m ,- m) }{\phi( X_{J} + \delta X_{K} \pm \imag \beta / 2 m ; m , -m)\phi( -X_{J} + \delta X_{K}\mp \imag \beta / 2 m; m , -m)}
= \prod_{\delta= \pm} \frac{s( X_{J} + \delta X_{K} \pm \imag \lambda m \beta / 2)}{s( X_{J} + \delta X_{k} \mp \imag \lambda m \beta / 2)},
\label{eq_proof_relation_3}\end{equation}
which have a common solution given by $\phi(x ; m , - m)$ in \eqref{eq_Sen_function_pair}.
To see this more clearly, we can consider the upper sign, that is `$-$' for `$\mp$' and `$+$' for `$\pm$', in \eqref{eq_proof_relation_3} and note that the right hand side can be written as
$$
\prod_{\delta= \pm} \frac{s( - X_{J} + \delta X_{K} - \imag \lambda m \beta / 2)}{s( X_{J} + \delta X_{k} - \imag \lambda m \beta / 2)}.
$$
Then it is a straightforward check to see that $\phi(x;m,-m)$ in \eqref{eq_Sen_function_pair} satisfies both difference equations.

When $(m_{J} , m_{K} ) = (m, +1/\lambda m)$ for $m\in\Lambda$, then \eqref{eq_proof_relation_1} and \eqref{eq_proof_relation_2}, using \eqref{eq_proof_constants}, become
\begin{multline*}
\prod_{\delta = \pm} \frac{\phi(X_{J} + \delta X_{K} \mp \imag \beta / 2 m ; m , 1/\lambda m)\phi(-X_{J} + \delta X_{K} \pm \imag \beta / 2m ; m , 1/ \lambda m)}{\phi(X_{J} + \delta X_{K} \pm \imag \beta / 2m; m , 1/\lambda m)\phi(-X_{J} + \delta X_{K} \mp \imag \beta / 2m; m , 1/ \lambda m)} \\
= \prod_{\delta = \pm} \frac{s( X_{J} + \delta X_{K} \pm( a(m) - a(1/\lambda m)) \mp \imag \beta /2 m)}{s( X_{J} + \delta X_{K} \pm( a(m) - a(1/\lambda m)) \pm \imag \beta / 2m )}
\end{multline*}
and
\begin{multline*}
\prod_{\delta = \pm} \frac{\phi( X_{J} + \delta X_{K} \mp \imag \lambda m \beta / 2 ; m , 1/\lambda m)\phi( -X_{J} + \delta X_{K} \pm \imag \lambda m \beta / 2 ; m , 1/\lambda m)}{\phi( X_{J} + \delta X_{K} \pm \imag \lambda m \beta / 2 ; m , 1/\lambda m)\phi( -X_{J} + \delta X_{K} \mp \imag \lambda m \beta / 2; m , 1/\lambda m)} \\
= \prod_{\delta = \pm} \frac{s(X_{J} + \delta X_{K} \pm ( a(1/\lambda m) - a(m)) \mp \imag \lambda m \beta/2)}{s(X_{J} + \delta X_{K} \pm ( a(1/\lambda m) - a(m))\pm \imag \lambda m \beta / 2)},
\end{multline*}
respectively. (The relations is obtained by shifting the $X_{J}$ variable by $\pm \imag \beta / 2m$ and changing $m\to +1/\lambda m$ in \eqref{eq_proof_relation_2}.)
The choice 
\begin{equation}
\label{eq_a_cond_2}a(m)-a(1/\lambda m)=0\quad(\forall m\in\Lambda)
\end{equation}
simplifies the four cases and allows us to find the common solution $\phi(x ; m , + 1/ \lambda m)$ in \eqref{eq_Sen_function_pair}.

When $(m_{J},m_{K}) = (m , - 1/ \lambda m)$ for $m\in\lambda$, then \eqref{eq_proof_relation_1} and \eqref{eq_proof_relation_2}, using \eqref{eq_proof_constants}, become
\begin{multline*}
\prod_{\delta = \pm} \frac{\phi(X_{J} + \delta X_{K} \mp \imag \beta / 2 m ; m , -1/ \lambda m) \phi( - X_{J} + \delta X_{K} \pm \imag \beta / 2 m ; m , -1/ \lambda m)}{\phi(X_{J} + \delta X_{K} \pm \imag \beta / 2 m ; m , -1/ \lambda m) \phi( -X_{J} + \delta X_{K} \mp \imag \beta / 2 m ; m , -1/ \lambda m)} \\
= \prod_{\delta= \pm} \frac{s(X_{J} + \delta X_{K} \pm( a(m) - a(-1/\lambda m)) \pm 3 \imag \beta / 2m)}{s(X_{J} + \delta X_{K} \pm( a(m) - a(-1/\lambda m)) \pm \imag \beta / 2 m)} \\ 
\cdot \Bigl(
\frac{s( X_{J} + \delta X_{K} \mp \imag \lambda m \beta/2 ) s( X_{J} + \delta X_{K} \pm \imag \lambda m \beta / 2)}{s(X_{J} + \delta X_{K} \mp \imag \lambda m \beta / 2 \pm \imag \beta / m) s( X_{J} + \delta X_{K} \pm \imag \lambda m \beta / 2 \mp \imag \beta / m)}
\Bigr)^{1/2}
\end{multline*}
and
\begin{multline*}
\prod_{\delta = \pm} \frac{\phi(X_{J} + \delta X_{K} \pm \imag \lambda m \beta /2 ; m , - 1/\lambda m)\phi(-X_{J} + \delta X_{K} \mp \imag \lambda m \beta / 2 ; m , -1/\lambda m)}{\phi(X_{J} + \delta X_{K} \mp \imag \lambda m \beta / 2 ; m , -1/\lambda m) \phi(-X_{J} + \delta X_{K}  \pm \imag \lambda m \beta / 2) } \\
= \prod_{\delta = \pm} \frac{s(X_{J} + \delta X_{K} \mp (a(m) - a(-1/\lambda m )) \mp 3 \imag \lambda m \beta / 2)}{s(X_{J} + \delta X_{K} \mp(a(m) - a(-1/\lambda m)) \mp \imag \lambda m \beta / 2)} \\
\cdot \Bigl( \frac{s(X_{J} + \delta X_{K} \pm \imag \beta / 2 m ) s( X_{J} + \delta X_{K}  \mp \imag \beta / 2m)}{s(X_{J} + \delta X_{K} \pm \imag \beta / 2m \mp \imag \lambda m \beta) s(X_{J} + \delta X_{K} \mp \imag \beta / 2m \pm \imag \lambda m \beta) }\Bigr)^{\half},
\end{multline*}
respectively. Choosing 
\begin{equation}
\label{eq_a_cond_3}
a(m)-a(-1/\lambda m) = - \imag \lambda m \beta /2 - \imag \beta / 2 m \quad (\forall m\in \Lambda)
\end{equation}
reduces the equations to
\begin{multline*}
\prod_{\delta = \pm} \frac{\phi(X_{J} + \delta X_{K} \mp \imag \beta / 2 m ; m , -1/ \lambda m) \phi( - X_{J} + \delta X_{K} \pm \imag \beta / 2 m ; m , -1/ \lambda m)}{\phi(X_{J} + \delta X_{K} \pm \imag \beta / 2 m ; m , -1/ \lambda m) \phi( -X_{J} + \delta X_{K} \mp \imag \beta / 2 m ; m , -1/ \lambda m)} \\
= \prod_{\delta= \pm} \Bigl(
\frac{s( X_{J} + \delta X_{K} \pm \imag \lambda m \beta/2 )s(X_{J} + \delta X_{K} \mp \imag \lambda m \beta / 2 \pm \imag \beta / m)}{s( X_{J} + \delta X_{K} \mp \imag \lambda m \beta / 2) s( X_{J} + \delta X_{K} \pm \imag \lambda m \beta / 2 \mp \imag \beta / m)}
\Bigr)^{\half}
\end{multline*}
and
\begin{multline*}
\prod_{\delta = \pm} \frac{\phi(X_{J} + \delta X_{K} \pm \imag \lambda m \beta /2 ; m , - 1/\lambda m)\phi(-X_{J} + \delta X_{K} \mp \imag \lambda m \beta / 2 ; m , -1/\lambda m)}{\phi(X_{J} + \delta X_{K} \mp \imag \lambda m \beta / 2 ; m , -1/\lambda m) \phi(-X_{J} + \delta X_{K}  \pm \imag \lambda m \beta / 2) } \\
= \prod_{\delta = \pm} \Bigl( \frac{ s( X_{J} + \delta X_{K}  \mp \imag \beta / 2m) s(X_{J} + \delta X_{K} \pm \imag \beta / 2m \mp \imag \lambda m \beta)}{ s(X_{J} + \delta X_{K} \pm \imag \beta / 2 m )  s(X_{J} + \delta X_{K} \mp \imag \beta / 2m \pm \imag \lambda m \beta) }\Bigr)^{\half},
\end{multline*}
which have a common solution given by $\phi(x ; m , -1 / \lambda m )$ in \eqref{eq_Sen_function_pair}.

We then find that the function 
\begin{equation}
a(m) = - \frac{\imag \lambda \beta}{4} ( m + 1/\lambda m) + a_{0}
\label{eq_a_function}
\end{equation}
satisfies the conditions \eqref{eq_a_cond_1}, \eqref{eq_a_cond_2}, and \eqref{eq_a_cond_3} for any constant $a_0 \in\C$. In the following we set this constant to zero, for simplicity.

Turning our attention to the $\psi(x ; m )$ functions, we find that it should satisfy
\begin{multline}\label{eq_proof_relation_4}
\Bigl( \frac{\prod_{\nu=0}^{2\rho +1 } s( \pm X + \imag \beta / 2m  - \imag d(g_{\nu},m) \beta ) s(\mp X_{J} + \imag \beta / 2 m - \imag d(g_{\nu},m) \beta)}{s( \pm 2  X + \imag \beta / m) s( \pm 2 X) s(\mp 2 X + \imag \beta/m ) s( \mp 2 X)} \Bigr)^{\half} \frac{\psi( X \mp \imag \beta/2 m ; m )}{\psi( X \pm \imag \beta / 2 m ; m)}\\
= C \prod_{\nu=0}^{\rho} \frac{s(\pm X + \imag \beta / 2 m - \imag \lambda m \beta / 2 - \half \omega_{\nu}) }{ s(\pm X + \imag \beta / 2 m - \half \omega_{\nu}) }\frac{s( \pm X + \imag \beta / 4 m - \imag \lambda m \beta /4 - c_{\nu} -\half \omega_{\nu} - n_{\nu})}{s( \pm X + \imag \beta / 4 m -\imag \lambda m \beta / 4 - c_{\nu} - \half \omega_{\nu}) } \\ \cdot \frac{ s( \pm X + \imag \beta / 4 m -\imag \lambda m \beta /4  - d_{\nu} - \half \omega_{\nu} - n_{\nu+\rho+1})}{s(\pm X + \imag \beta / 4 m  - \imag \lambda m \beta / 4 -d_{\nu} - \half \omega_{r}) }
\end{multline}
for some constant $C\in\C^{\ast}$ and all $m \in \Lambda$. (In the formula we replaced $m_{J}$ with $m$, $X_{J}$ with $X \pm \imag \beta / 2 m$ and $a_{J} = a(m_{J})$ in \eqref{eq_a_function} with $a_{0}=0$.) Choosing
$$
c_{\nu} = \frac{\imag \beta}{4}(\lambda -1) , \quad d_{\nu} = - \frac{\imag \beta}{4}(\lambda -1 ) \quad (\nu=0,\ldots,\rho)
$$
we find that 
$$
+ \imag \beta / 2 m - \imag \lambda m \beta / 2 = \begin{cases} + \imag \beta / 4 m -\imag \lambda m \beta / 4 - c_{\nu} \quad &\text{if } m=1 \text{ or } m=-1/\lambda \\
+ \imag \beta / 4 m -\imag \lambda m \beta /4  - d_{\nu} \quad &\text{if } m= -1 \text{ or } m= 1/\lambda 
\end{cases}
$$
and that the r.h.s. of \eqref{eq_proof_relation_4} can be written as
$$
C  \prod_{\nu=0}^{\rho} \frac{s(\pm X - \imag \quarter \lambda [m+1] \beta + \imag\quarter [1 + 1/m] \beta - \half\omega_{\nu} - n_{\nu}) s(\pm X - \imag \quarter \lambda [ m - 1]\beta - \imag \quarter [1 - 1/m] \beta- \half \omega_{\nu} - n_{\nu+\rho+1}) }{s(\pm X - \half\omega_{\nu}) s(\pm X - \half\omega_{\nu} + \imag \beta / 2 m)}.
$$
(Note that if the constant $a_{0}$ in \eqref{eq_a_function} is non-zero, then this result still holds if we shift $c_{\nu}$ and $d_{\nu}$ by $a_{0}$.) Recalling $d(g,m)$ in \eqref{eq_Sen_parameters} and noting that the duplication formula for $s(x)$ can be expressed as
\begin{equation}\label{eq_double_angle_formula}
s( 2 x ) = 2 \prod_{\nu=0}^{\rho} s(x - \half \omega_{\nu}) / \prod_{\nu=1}^{\rho} s(-\half \omega_{\nu}),
\end{equation}
we can express the square-root term in the l.h.s. of \eqref{eq_proof_relation_4} as
\begin{multline*}
\Bigl( \frac{\prod_{\nu=0}^{2\rho +1 } s( \pm X + \imag \beta / 2m  - \imag d(g_{\nu},m) \beta ) s(\mp X_{J} + \imag \beta / 2 m - \imag d(g_{\nu},m) \beta)}{s( \pm 2  X + \imag \beta / m) s( \pm 2 X) s(\mp 2 X + \imag \beta/m ) s( \mp 2 X)} \Bigr)^{\half}
= \frac{\prod_{\nu=1}^{\rho} s(\half \omega_{\nu})^2}{4}  \\ \cdot \Bigl( \frac{\prod_{\nu=0}^{2\rho +1}s( \pm X - \imag \quarter \lambda [m-1] \beta + \imag \quarter [1 + 1/m] \beta - \imag g_{\nu} \beta)s( \mp X  - \imag \quarter \lambda [m-1] \beta + \imag \quarter [1 + 1/m] \beta - \imag g_{\nu} \beta)}{\prod_{\nu=0}^{\rho} s(\pm X + \imag \beta /2 m - \half \omega_{\nu})s(\pm X - \half \omega_{\nu}) s(\mp X + \imag \beta / 2 m - \half \omega_{\nu}) s(\mp X - \half \omega_{\nu})}
\Bigr)^{\half}.
\end{multline*}
We then find that choosing $C =(1/4) \prod_{\nu=1}^{\rho}s(\omega_{\nu}/2)^2$ and 
$$
n_{\nu} = -\imag \lambda \beta / 2 - \half \omega_{\nu} + \imag g_{\nu} \beta , \quad n_{\nu+\rho+1} = -\imag \beta / 2 - \half \omega_{\nu} + \imag g_{\nu+\rho+1} \beta \quad (\nu=0,\ldots , \rho)
$$
simplifies the relation in \eqref{eq_proof_relation_4} to
\begin{multline*}
\frac{\psi( X \mp \imag \beta/2 m ; m )}{\psi( X \pm \imag \beta / 2 m ; m)} = \Bigl( \prod_{\nu=0}^{\rho} \frac{s(\mp X -\half \omega_{\nu}) s(\mp X + \imag \beta / 2 m -\half \omega_{\nu})}{s( \pm X - \half \omega_{\nu}) s(\pm X + \imag \beta / 2 m -\half \omega_{\nu})} \Bigr. \\ \cdot
\Bigl. \prod_{\nu=0}^{2\rho+1} \frac{s(\pm X - \imag \quarter \lambda [ m - 1] \beta + \imag \quarter [1 + 1/m] \beta - \imag g_{\nu} \beta)}{s(\mp X - \imag \quarter \lambda[m-1] \beta + \imag \quarter [1+1/m] \beta - \imag g_{\nu} \beta)} \Bigr)^{\half}
\\
= \Bigl( \frac{s(\mp 2 X ) s(\mp 2 X + \imag \beta / m)}{s(\pm 2 X ) s(\pm 2 X + \imag \beta / m)} \prod_{\nu=0}^{2\rho+1} \frac{s(\pm X + \imag \beta / 2 m- \imag d(g_{\nu},m) \beta )}{s(\mp X +\imag \beta / 2m - \imag d(g_{\nu},m)\beta )} \Bigr)^{\half}
\end{multline*}
which have the common solution $\psi(x ; m)$ given in \eqref{eq_Sen_function_single}. (The last line is obtained by using \eqref{eq_Sen_parameters} and \eqref{eq_double_angle_formula}.)

Combining the results above, we use Lemma~\ref{lemma_key_lemma} in order to express
\begin{multline}\label{eq_proof_relation_5}
\cH_{\cN}(\bs X ; \bs m) \Phi(\bs X ; \bs m) = \Bigl[ \cV^{0}(\bs X ; \bs m) + \frac{\prod_{\nu=1}^{\rho}s(\half \omega_{\nu})^{2} }{4} \Bigl( s( 2 \imag \lambda \beta \sum_{J=1}^{\cN} m_{J} + \sum_{\nu=0}^{\rho}[ \imag g_{\nu} \beta + \imag g_{\nu+\rho+1} \beta- \imag (\lambda + 1 ) \beta / 2 - \omega_{\nu}]) \Bigr. \Bigr. 
\\
+ \sum_{\nu=0}^{\rho} \bigl\{ \bigr. \frac{\prod_{\mu=0}^{\rho} s( \half \omega_{\nu} + \imag \lambda \beta /2 - \imag g_{\mu} \beta) }{\prod_{\mu\neq\nu}^{\rho}s( \half  ( \omega_{\nu} - \omega_{\mu}))} (\prod_{\mu=0}^{\rho}\frac{s(\half \omega_{\nu} + \imag \lambda \beta / 2 - \imag g_{\mu+\rho+1} \beta) }{s(\half( \omega_{\nu} - \omega_{\mu} + \imag (\lambda - 1 ) \beta))} )
\\ \cdot
\prod_{\delta = \pm}	\prod_{J=1}^{N} \frac{s(\delta X_{J} + \half \omega_{\nu} + \imag \quarter \lambda [m_{J} +1] \beta - \imag \quarter [ 1 - 1/m_{J}] \beta - \imag \lambda m_{J} \beta) }{s(\delta X_{J} + \half \omega_{\nu} + \imag \quarter \lambda [m_{J} +1] \beta - \imag \quarter [ 1 - 1/m_{J}] \beta)} \\
	+ 	
	(\prod_{\mu=0}^{\rho} \frac{s(\half \omega_{\nu} + \imag \beta / 2 - \imag g_{\mu} \beta)}{s(\half ( \omega_{\nu} - \omega_{\mu} +(1-\lambda)\beta ))}) \frac{\prod_{\mu=0}^{\rho}s(\half \omega_{\nu} + \imag \beta /2 - \imag g_{\mu+\rho+1} \beta) }{ \prod_{\mu\neq\nu}^{\rho}s(\half ( \omega_{\nu}- \omega_{\mu}))} \\
	\cdot \Bigl. \Bigl.  \prod_{\delta = \pm} \prod_{J=1}^{\cN} \frac{s(\delta X_{J} + \half \omega_{\nu} +\imag \quarter \lambda [m_{J} -1] \beta + \imag \quarter [1 + 1/m_{J}] \beta - \imag \lambda m_{J} \beta) }{s(\delta X_{J} + \half \omega_{\nu} + \imag \quarter \lambda [m_{J} -1] \beta + \imag \quarter [1 + 1/m_{J}] \beta) } \bigl. \bigr\} \Bigr) \Bigr] \Phi( \bs X ; \bs m)
\end{multline}
in the rational (I), trigonometric (II), and hyperbolic (III) cases, while in the elliptic case \eqref{eq_proof_relation_5} holds only if the parameters satisfy 
$$
2 \lambda  \sum_{J=1}^{\cN} m_{J} +   \sum_{\nu=0}^{7} g_{\nu}  - 2   ( \lambda + 1 ) = 0 \quad \text{(IV)}.
$$
Recalling the coefficients $\cV^{0}$ \eqref{eq_Sen_coeff_potential}, it becomes clear that the $\bs{X}$-dependent part in the square brackets cancel and we obtain the eigenvalue identity in Theorem~\ref{thm_source_identity}.

To summarize: We showed that \eqref{eq_proof_eigenvalue_identity}, with $\psi(x;m)$ in \eqref{eq_Sen_function_single} and $\phi(x ; m , m')$ in \eqref{eq_Sen_function_pair}, is identical with \eqref{eq_Key_eq} multiplied by an overall factor $(1/4)\prod_{\nu=1}^{\rho} s(\omega_{\nu} /2)^2$ for constants
\begin{equation}\label{eq_proof_constants_2}
\begin{aligned}
c_{\nu} = a_{0} + \frac{\imag \beta}{4}(\lambda -1 ), &\quad d_{\nu} = a_{0} - \frac{\imag \beta}{4}(\lambda -1) , \\ n_{\nu} = - \frac{\imag \lambda \beta}{2} - \frac{\omega_{\nu}}{2} + \imag g_{\nu} \beta , &\quad n_{\nu+\rho+1} = -\frac{\imag \beta}{2} - \frac{\omega_{\nu}}{2} + \imag g_{\nu + \rho +1} \beta,
\end{aligned} \quad (\nu=0,\ldots , \rho)
\end{equation}
and 
$$
a_{J} = - \frac{\imag \lambda \beta }{4}( m_{J} + 1/ \lambda m_{J}) + a_{0} \quad ( J=1,\ldots , \cN)
$$
where $a_{0} \in \C$ is any arbitrary constant. 
Note that the balancing condition \eqref{eq_balancing_condition} is the same as \eqref{eq_lemma_balancing_condition} with parameters in \eqref{eq_proof_constants_2}.
\end{proof}

\subsection{Proof of Lemma~\ref{lemma_Sen_conjugate_operator}}\label{section_proof_of_lemma}
\begin{proof}
From our proof above it is straightforward to deduce that
\begin{multline*}
\Phi( \bs X ; \bs m)^{-1} \circ \cH_{\cN}(\bs X ; \bs m ) \circ \Phi(\bs X ; \bs m) = \cV^{0}(\bs X ; \bs m) + \sum_{\ve = \pm} \sum_{J=1}^{\cN} s( \imag \lambda m_{J} \beta) 
\\ \cdot
 \Bigl( \frac{\prod_{\nu=0}^{2\rho +1 } s( \ve X_{J} - \imag d_{\nu,J} \beta ) s(-\ve X_{J} + \imag \beta /m_{J} - \imag d_{\nu,J}\beta)}{s(2\ve X_{J}) s(2 \ve X_{J} - \imag \beta/ m_{J}) s(-2 \ve X_{J} + 2 \imag \beta/m_{J}) s(-2\ve X_{J} + \imag \beta / m_{J})} \Bigr)^{\half} \frac{\psi( X_{J} - \ve \imag \beta/m_{J} ; m_{J})}{\psi( X_{J} ; m_{J})}
\\ \cdot
\prod_{\delta = \pm} \prod_{K\neq J}^{\cN} f_{\ve}(X_{J} + \delta X_{K} ; m_{J} ,m_{K}) f_{-\ve}(X_{J} + \delta X_{K} - \ve \imag \beta / m_{J} ; m_{J} , m_{K}) \\
\prod_{K \less J} \frac{\phi( X_{J} + \delta X_{K} - \ve \imag \beta / m_{J} ; m_{J},m_{K})\phi( -X_{J} + \delta X_{K} + \ve \imag \beta / m_{J} ; m_{J},m_{K})}{\phi( X_{J} + \delta X_{K} ; m_{J} , m_{K}) \phi( X_{J} + \delta X_{K} ; m_{J},m_{K})} 
\\ \cdot
\prod_{K \gtr J} \frac{\phi( X_{J} + \delta X_{K} - \ve \imag \beta / m_{J} ; m_{K},m_{J})\phi( -X_{J} + \delta X_{K} + \ve \imag \beta / m_{J} ; m_{K},m_{J})}{\phi( X_{J} + \delta X_{K} ; m_{K} , m_{J}) \phi( X_{J} + \delta X_{K} ; m_{K},m_{J})} \exp\Bigl( - \ve \imag \frac{\beta}{m_{J}}\frac{\partial}{\partial X_{J}}\Bigr)
\end{multline*}
equals
\begin{multline*}
\cV^{0}(\bs X ; \bs m) + \sum_{\ve = \pm} \sum_{J=1}^{\cN} s( \imag \lambda m_{J} \beta) 
\frac{\prod_{\nu=0}^{2\rho +1 } s( \ve X_{J} - \imag d_{\nu,J} \beta )}{s(2\ve X_{J}) s(2 \ve X_{J} - \imag \beta/ m_{J})}
\\ \cdot
 \prod_{\delta = \pm} \prod_{K\neq J}^{\cN}  \frac{s(X_{J} + \delta X_{K} + \ve[ a(m_{J}) - a(m_{K})] - \ve\imag \lambda m_{K} \beta)}{s(X_{J} + \delta X_{K} + \ve[ a(m_{J}) - a(m_{K})] )} \exp\Bigl( - \ve \imag \frac{\beta}{m_{J}}\frac{\partial}{\partial X_{J}}\Bigr),
\end{multline*}
with $a(m)$ in \eqref{eq_a_function}. Recalling the function $f_{\pm}(x ; m ,m')$ in \eqref{eq_Sen_coeff_pair} shows that this term equals \eqref{eq_lemma_1}. The last relation is obtained from Lemma~\ref{lemma_key_lemma} using \eqref{eq_proof_constants_2} which gives that 
\begin{multline*}
\sum_{\ve = \pm} \sum_{J=1}^{\cN} s( \imag \lambda m_{J} \beta) 
\frac{\prod_{\nu=0}^{2\rho +1 } s( \ve X_{J} - \imag d_{\nu,J} \beta )}{s(2\ve X_{J}) s(2 \ve X_{J} - \imag \beta/ m_{J})}
 \prod_{\delta = \pm} \prod_{K\neq J}^{\cN} \frac{s(X_{J} + \delta X_{K} + \ve[ a(m_{J}) - a(m_{K})] - \ve\imag \lambda m_{K} \beta)}{s(X_{J} + \delta X_{K} + \ve[ a(m_{J}) - a(m_{K})] )} \\
+ \cV^{0}(\bs X ; \bs m)  = \frac{\prod_{\nu=1}^{\rho}s(\half \omega_{\nu})}{4} s( \imag \beta [ 2 \lambda \sum_{J=1}^{\cN} m_{J} + \sum_{\nu=0}^{2\rho +1} g_{\nu} - \frac{1}{2} (\rho + 1)( \lambda + 1) ] - \sum_{\nu=0}^{\rho} \omega_{\nu})
\end{multline*}
in the rational (I), trigonometric (II), and hyperbolic (III) cases, while in the elliptic (IV) case it holds if and only if the parameters satisfy \eqref{eq_balancing_condition}.
\end{proof}
\section{Special cases}\label{section_special_cases}

It is worth noting that there are several interesting eigenvalue identities and kernel function identities that can be obtained from the source identity in Theorem~\ref{thm_source_identity}.

In the special case where $\cN = N \in\Z_{\gtr 0}$ and $(m_{J},X_{J}) = (1 , x_{J})$, for all $J = 1 ,\ldots , N$, we find that the operator in \eqref{eq_Sen_operator} reduces to $H_{N}(\bs x ; \bs g ,\lambda , \beta ) - c^{0}$ where $H_{N}$ is the Koornwinder-van Diejen operator, that is
$$
H_{N}(\bs x ; \bs g , \lambda , \beta ) = s(\imag g \beta ) \sum_{\ve=\pm} \sum_{j=1}^{N} V_{j}^{\ve}(\bs x ; \bs g , \lambda ,\beta )^{1/2} 
 \e^{- \ve \imag \beta \frac{\partial}{\partial x_{j}}} V_{j}^{-\ve}(\bs x ; \bs g , \lambda ,\beta )^{1/2}  + V^{0}(\bs x;\bs g , \lambda , \beta) 
$$
with coefficients
$$
V^{\pm}_{j}(\bs x ; \bs g , \lambda , \beta ) =\frac{\prod_{\nu=0}^{2\rho +1 } s(\pm x_{j} - \imag g_{\nu}\beta)}{s(\pm2  x_{j}) s(\pm2 x_{j} - \imag \beta)} \prod_{\delta = \pm} \prod_{k\neq j}^{N} \frac{s( x_{j} + \delta x_{k} \mp \imag \lambda \beta)}{s(x_{j} + \delta x_{k})} 
$$
and
\begin{multline}
V^{0}(\bs x ; \bs g , \lambda , \beta ) = - \frac{\prod_{\nu= 1}^{\rho} s(\half\omega_{\nu})^2}{4}\sum_{\nu=0}^{\rho} \frac{\e^{- r  \xi_{\nu} ( 2\lambda N + \abs{\bs g} - \half (\rho + 1 ) ( \lambda + 1))\beta}}{\prod_{\mu\neq \nu}^{\rho}s(\half (\omega_{\nu}-\omega_{\mu}))} \frac{\prod_{\mu=0}^{2\rho +1 } s( \half \omega_{\nu} + \imag \beta / 2 - \imag g_{\mu} \beta)}{\prod_{\mu=0}^{\rho} s(\half (\omega_{\nu} - \omega_{\mu} + \imag ( 1 - \lambda ) \beta )) }\\ \cdot
 \prod_{\delta = \pm} \prod_{j=1}^{N} \frac{s( \delta x_{j} + \half \omega_{\nu} + \imag \beta /2 - \imag \lambda \beta )}{s(\delta x_{j} + \half \omega_{\nu} + \imag \beta / 2)},
\label{eq_van_Diejen_potential}
\end{multline}
and $c^{0} \in \C$ is given by
\begin{equation}
c^{0} = \frac{\prod_{\nu= 1}^{\rho} s(\half\omega_{\nu})^2}{4}\sum_{\nu=0}^{\rho} \frac{\e^{- r  \xi_{\nu} (\abs{\bs g} - \half (\rho + 1 ) ( \lambda + 1))\beta}}{\prod_{\mu\neq \nu}^{\rho}s(\half (\omega_{\nu}-\omega_{\mu}))} \frac{\prod_{\mu=0}^{2\rho +1 } s( \half \omega_{\nu} + \imag \lambda  \beta / 2 - \imag g_{\mu} \beta)}{\prod_{\mu=0}^{\rho} s(\half (\omega_{\nu} - \omega_{\mu} + \imag ( \lambda -  1) \beta )) }.
\label{eq_GS_constant}\end{equation}
The source identity \eqref{eq_Source_identity} reduces then to the following eigenvalue equation:
\begin{corollary}\label{cor_vD_GS}
The function
\begin{multline}\label{eq_GS_eigenfunction}
\Psi_{N}(\bs x ; \bs g , \lambda , \beta ) = \prod_{j=1}^{N} \Bigl(\frac{G(2 x_{j} + \imag \beta / 2 ; \beta)G(-2 x_{j} + \imag \beta / 2 ; \beta)}{\prod_{\nu=0}^{2\rho + 1} G(x_{j} + \imag \beta / 2 - \imag g_{\nu}\beta; \beta) G( - x_{j} + \imag \beta / 2 - \imag g_{\nu} \beta ; \beta)} \Bigr)^{1/2} \\ \cdot
 \prod_{1 \leq j \less k\leq N} \prod_{\ve , \ve'=\pm} \Bigl( \frac{G( \ve x_{j} + \ve' x_{k} + \imag \beta / 2 ; \beta)}{G(\ve x_{j} + \ve' x_{k}- \imag \lambda \beta +\imag \beta / 2 ; \beta)} \Bigr)^{1/2}
\end{multline}
satisfies the eigenvalue equation
$$
\bigl(H_{N}(\bs x; \bs g , \lambda , \beta ) - \cE_{N}(\bs g , \lambda , \beta) \bigr) \psi_{N}(\bs x ; \bs g , \lambda , \beta ) = 0
$$
for
$$
\cE_{N}(\bs g , \lambda , \beta) = c^{0} + \frac{\prod_{\nu=1}^{\rho}s(\half \omega_{\nu})^{2}}{4} s( \imag \beta [ 2 \lambda N + \abs{\bs g} - \frac{1}{2}(\rho+1)(\lambda + 1 )] - \abs{\bs\omega}) ,
$$
with $c^{0}$ in \eqref{eq_GS_constant}, in the rational \emph{(I)}, trigonometric \emph{(II)}, and hyperbolic \emph{(III)} cases for all parameters, and for 
$$
2 \lambda (N - 1) + \sum_{\nu=0}^{7} g_{\nu} - 2 = 0 \quad \text{\emph{(IV)}}
$$
in the elliptic \emph{(IV) }case.
\end{corollary}
The function $\Psi_{N}(\bs x ; \bs g , \lambda , \beta)$ is the groundstate eigenfunction for the operator $H_{N}(\bs x ; \bs g , \lambda , \beta)$ and by Lemma~\ref{lemma_Sen_conjugate_operator}, it is the square-root of the weight function with respect to which the Koornwinder-van Diejen operators $A_{N}(\bs x ; \bs g , \lambda ,\beta )$ are symmetric.

In the special case where $\cN = N + M$ ($N,M\in\Z_{\gtr 0}$) and 
\begin{equation}\label{eq_vD_KF_case}
(m_{J} , X_{J} ) = \begin{dcases} ( 1 , x_{J} ) \quad &\text{for } J=1 ,\ldots N \\
(-1 , y_{J-N})  \quad &\text{for } J-N=1 ,\ldots M
\end{dcases},
\end{equation}
we find that the operator in \eqref{eq_Sen_operator} reduces to a sum of two Koornwinder-van Diejen type operators. The source identity in this case yields a well-known kernel function identity for two pairs of Koornwinder-van Diejen type difference operators:
\begin{corollary}\label{cor_vD_Cauchy_KFI}
The function
\begin{equation}\label{eq_vD_KF}
F_{N,M}(\bs x ,\bs y ; \bs g , \lambda , \beta ) = \Psi_{N}(\bs x ; \bs g , \lambda \beta ) \Psi_{M}(\bs y ; \bs{\tilde{g}} , \lambda,  \beta) \prod_{j=1}^{N} \prod_{k=1}^{M} \prod_{\ve , \ve' = \pm} G(\ve x_{j} + \ve' y_{k} - \imag \lambda \beta /2 ; \beta),
\end{equation}
where $\tilde{g}_{\nu} = (\lambda + 1 )/2 - g_{\nu}$ $(\nu=0,\ldots, 2\rho +1)$, satisfies
$$
\Bigl( H_{N}(\bs x ; \bs g , \lambda , \beta ) - H_{M}(\bs y ; \bs{\tilde{g}} , \lambda , \beta ) - C_{N,M} \Bigr) F_{N,M}(\bs x , \bs y ; \bs g , \lambda ,\beta) = 0
$$
with
$$
C_{N,M} = \frac{\prod_{\nu= 1}^{\rho} s(\half\omega_{\nu})^2}{4} s( \imag \beta [ 2 \lambda ( N - M ) + \abs{\bs{g}} - \frac{1}{2} (\rho +1 ) ( \lambda + 1) ] - \abs{\bs\omega}),
$$
for arbitrary parameters in the rational \emph{(I)}, trigonometric \emph{(II)}, and hyperbolic \emph{(III)} cases, and for 
$$
2 \lambda (N - M - 1) + \sum_{\nu=0}^{7} g_{\nu} - 2 = 0 \quad \text{\emph{(IV)}}
$$
in the elliptic case.
\end{corollary}
\begin{proof}
It is clear that the function in \eqref{eq_Sen_function} reduces to \eqref{eq_vD_KF} in the special case \eqref{eq_vD_KF_case} so we need only to consider the operator \eqref{eq_Sen_operator} in this special case. Noting that the function $f_{\pm}$ satisfy $$f_{\pm}(x ; m , m') f_{\mp}(x \mp \imag \beta / m ; m , m') =1$$ when $m'=-m$, it becomes clear that the coefficients $\cV_{J}^{\pm} (\bs X ; \bs m)$ become
\begin{multline}
\begin{dcases}
\frac{\prod_{\nu=0}^{2 \rho+1} s( \pm x_{J} - \imag g_{\nu} \beta )}{s(\pm 2 x_{J}) s(\pm 2 x_{J}- \imag \beta)} \prod_{\delta = \pm} \prod_{\substack{K=1 \\ K \neq J}}^{N} \frac{s(x_{J} + \delta x_{K} \mp \imag \lambda \beta)}{s(x_{J} + \delta x_{K})} \quad &\text{for } J=1,\ldots,N \\
\frac{\prod_{\nu=0}^{2\rho+1} s(\pm y_{J-N} - \imag (g_{\nu} - (\lambda + 1) /2 )\beta)}{s(\pm 2 y_{J-N}) s( \pm 2 y_{J-N} + \imag \beta )}  \prod_{\delta = \pm} \prod_{\substack{ K = 1 \\ K \neq J-N}}^{M} \frac{s( y_{J-N} + \delta y_{K} \pm \imag \lambda \beta)}{s( y_{J-N} + \delta y_{K})}\quad &\text{for } J-N=1,\ldots,M
\end{dcases},
\end{multline}
and the coefficients $\cV^{0}$ become
\begin{multline*}
- \frac{\prod_{\nu=1}^{\rho} s(\half \omega_{\nu})^{2}}{4} \sum_{\nu=0}^{\rho} \frac{\e^{-r \xi_{\nu}( 2 \lambda  (N-M) + \abs{\bs{g}} - (\rho+1)(\lambda+1)/2)\beta}}{\prod_{\mu \neq \nu}^{\rho} s(\half(\omega_{\nu} - \omega_{\mu}))}\\ \cdot
\Bigl( \e^{- 2 r \xi_{\nu} M  \lambda \beta}  \frac{\prod_{\mu=0}^{2 \rho + 1} s(\half \omega_{\nu} + \imag \beta / 2 - \imag g_{\nu} \beta)}{\prod_{\mu=0}^{\rho} s(\half (\omega_{\nu} - \omega_{\mu} + \imag (1 - \lambda ) \beta))} \prod_{\delta = \pm} \prod_{j=1}^{N} \frac{s(\delta x_{j} + \half \omega_{\nu} + \imag \beta / 2 - \imag \lambda \beta)}{s(\delta x_{j} + \half \omega_{\nu} + \imag \beta / 2)} 
\Bigr. 
\\ \Bigl. +  \e^{+ r \xi_{\nu} 2 N   \lambda \beta} \frac{\prod_{\nu=0}^{2\rho + 1}s (\half \omega_{\nu} + \imag \lambda \beta / 2 - \imag g_{\nu} \beta)}{\prod_{\mu=0}^{\rho} s(\half ( \omega_{\nu} - \omega_{\mu} + \imag (\lambda - 1) \beta ))}  \prod_{\delta = \pm} \prod_{k=1}^{M} \frac{s(\delta y_{k} + \half \omega_{\nu} - \imag \beta / 2 + \imag \lambda \beta)}{s(\delta y_{k} + \half \omega_{\nu} - \imag \beta /2)} 
\Bigr)
\end{multline*}
in the case of \eqref{eq_vD_KF_case}. The new exponential factors comes from the terms 
$$
\e^{-  r \xi_{\nu} 2 M  \lambda \beta} = \prod_{\delta = \pm} \prod_{k=1}^{M} \frac{s(\delta y_{k} + \half \omega_{\nu} + \imag \lambda \beta / 2) }{s(\delta y_{k}+ \half \omega_{\nu} - \imag \lambda \beta / 2)}
$$
and
$$
 \e^{+ r \xi_{\nu} 2 N   \lambda \beta}  = \prod_{\delta = \pm} \prod_{j=1}^{N} \frac{s(\delta x_{j}+ \half \omega_{\nu} - \imag \lambda \beta / 2 )}{s(\delta x_{j} + \half \omega_{\nu} + \imag \lambda \beta / 2)}
$$
by using the properties of the $s$-function, that is using that $s(x)$ is an odd function and the quasi-periodicity in \eqref{eq_quasi_periodicity}. From this we see that the source identity \eqref{eq_Source_identity} becomes
$$
(H_{N}(\bs x ;\bs g , \lambda , \beta) + H_{M}(\bs y ; - \bs{\tilde{g}} , \lambda , -\beta) - C_{N,M}) F_{N,M}(\bs x , \bs y ; \bs g , \lambda , \beta)= 0,
$$
where we use the notation $-\bs{\tilde{g}} = (- \tilde{g}_{0},\ldots,-\tilde{g}_{2\rho +1})$. The final observation is that the Koornwinder-van Diejen operators satisfy $H_{N}(\bs x ; \bs g , \lambda , \beta ) = - H_{N}(\bs x ; -\bs{g},\lambda , -\beta)$ which concludes the proof. (Note that the this anti-symmetry differs from the normal symmetry of the Koornwinder-van Diejen operators due to our different normalization.)
\end{proof}
The kernel function identity in Corollary~\ref{cor_vD_Cauchy_KFI} recovers known results due to Ruijsenaars \cite{Rui09} for $N=M$ and Komori, Noumi, and Shiraishi \cite{KNS09} in the general $N,M$ case.

In the special case where $\cN = N + \tM$ ($N , \tM \in\Z_{\gtr 0}$) and
\begin{equation}
(m_{J} , X_{J}) = \begin{dcases}
(1 , x_{J}) \quad&\text{for } J = 1 ,\ldots N \\
(+1/\lambda , \ty_{J-N}) \quad&\text{for } J-N = 1 ,\ldots \tM
\end{dcases},
\label{eq_vD_dual_Cauchy_case}
\end{equation}
we again find that the operator in \eqref{eq_Sen_operator} becomes the sum of two Koornwinder-van Diejen type operators and the source identity yields another known kernel function identity for two pairs of Koornwinder-van Diejen type difference operators
\begin{corollary}\label{cor_vD_dual_Cauchy_KFI}
The function 
$$
\tilde{F}_{N , \tM}(\bs x , \bs\ty ;  \bs g ,\lambda ,\beta) = \Psi_{N}(\bs x ; \bs g , \lambda , \beta) \Psi_{M}(\bs\ty ;  \lambda^{-1} \bs g , 1/ \lambda , \lambda \beta) \prod_{j=1}^{N} \prod_{\delta = \pm} \prod_{k=1}^{\tM} s(x_{j} + \delta \ty_{k}),
$$
where $ \lambda^{-1}\bs g $ is a short-hand for $(g_{0}/ \lambda , \ldots , g_{2 \rho + 1 }/\lambda)$, satisfies 
$$
\Bigl( H_{N}(\bs x ; \bs g , \lambda , \beta ) + H_{M}(\bs\ty ; \lambda^{-1} \bs g , 1/\lambda , \lambda \beta ) - \tilde{C}_{N,\tM}\Bigr) \tilde{F}_{N,\tM}(\bs x ,\bs\ty ; \bs g , \lambda , \beta) = 0
$$
with
$$
\tilde{C}_{N,\tM} = \frac{\prod_{\nu=1}^{\rho} s(\half\omega_{\nu})^{2}}{4} s( \imag \beta [ 2 \lambda N + 2 \tM + \abs{\bs g} - \frac{1}{2} ( \rho + 1 ) (\lambda + 1)] - \abs{\bs\omega}),
$$
for arbitrary parameters in the rational \emph{(I)}, trigonometric \emph{(II)}, and hyperbolic \emph{(III)} cases, and for 
$$
2 \lambda (N - 1) + 2( \tM - 1) + \sum_{\nu=0}^{7} g_{\nu} = 0 \quad \text{\emph{(IV)}}
$$
in the elliptic \emph{(IV)} case.
\end{corollary}

In the trigonometric (II) case, Corollary~\ref{cor_vD_dual_Cauchy_KFI} recovers the result of Mimachi \cite{Mim01} while the elliptic generalization of Mimachi's result was given in \cite{KNS09}. 

In the special case where $\cN = N + \tN$ ($N , \tN$) and $\bs X ,\bs m$ in \eqref{eq_deformed_case}, we find that the operator in \eqref{eq_Sen_operator} reduces to a CFSV type generalization of the Koornwinder-van Diejen operators. More specifically, \eqref{eq_Sen_operator} becomes $H_{N,\tN}(\bs x , \bs \tx ; \bs g ,\lambda ,\beta) - c^{0}$ where 
\begin{equation*}
\begin{aligned}
H_{N,\tN}(\bs x , \bs\tx ; \bs g , \lambda ,\beta) =& \sum_{\ve=\pm}  s(\imag \lambda \beta ) \sum_{j=1}^{N} V_{j}^{\ve}(\bs x , \bs \tx)^{1/2} \exp\Bigl({-\ve\imag \beta \frac{\partial}{\partial x_{j}}}\Bigr) V_{j}^{-\ve}(\bs x , \bs \tx)^{1/2} \\
-& s( \imag \beta ) \sum_{k=1}^{\tN} \tilde{V}_{k}^{\ve}(\bs x , \bs \tx )^{1/2} \exp\Bigl({ + \ve \imag \lambda \beta \frac{\partial}{\partial \tx_{k}}}\Bigr) \tilde{V}_{k}^{-\ve}(\bs x , \bs \tx )^{1/2} + V^{0}(\bs x , \bs \tx),
\end{aligned}
\end{equation*}
with coefficients
$$
V^{\pm}_{j} = \frac{\prod_{\nu=0}^{2\rho+1} s( \pm x_{j} - \imag g_{\nu} \beta)}{s(\pm 2 x_{j} ) s(\pm 2 x_{j} - \imag \beta)} \prod_{\delta=\pm} \prod_{j' \neq j}^{N} \frac{s( x_{j} + \delta x_{j'} \mp \imag \lambda \beta)}{s( x_{j} + \delta x_{j'})} \prod_{k=1}^{\tN} \frac{s( x_{j} + \delta \tx_{k} \mp \imag ( \lambda - 1 ) \beta / 2) }{s( x_{j} + \delta \tx_{k} \mp \imag ( \lambda + 1 ) \beta / 2) },
$$
$$
\tilde{V}_{k}^{\pm} = \frac{\prod_{\nu=0}^{2 \rho +1 } s( \pm \tx_{k} + \imag \tilde{g}_{\nu} \beta )}{s(\pm 2 \tx_{k}) s(\pm 2 \tx_{k} + \imag \lambda \beta)} \prod_{\delta = \pm} \prod_{k' \neq k}^{\tN} \frac{s(\tx_{k} + \delta \tx_{k'} \pm \imag \beta)}{s(\tx_{k} + \delta \tx_{k'})} \prod_{j=1}^{N} \frac{s( \tx_{k} + \delta x_{j} \mp \imag (\lambda - 1 ) \beta / 2) }{s( \tx_{k} + \delta x_{j} \pm \imag (\lambda + 1 ) \beta / 2) },
$$
where $\tilde{g}_{\nu} = (\lambda + 1)/2 - g_{\nu}$ for all $\nu=0,\ldots,2\rho+1$, and
\begin{multline}\label{eq_deformed_operator_potential}
V^{0}(\bs x , \bs\tx) = - \frac{\prod_{\nu=1}^{\rho} s(\half \omega_{\nu})^{2}}{4} \sum_{\nu=0}^{\rho} \frac{\e^{- r \xi_{\nu} ( 2 \lambda N - 2 \tN + \abs{\bs{g}} - \half (\rho + 1 ) (\lambda + 1) ) \beta }}{\prod_{\mu \neq \nu}^{\rho} s(\half(\omega_{\nu} - \omega_{\mu}))} \frac{\prod_{\mu=0}^{2 \rho + 1}s(\half \omega_{\nu} + \imag \beta /2 - \imag g_{\mu} \beta )}{\prod_{\mu=0}^{\rho} s(\half ( \omega_{\nu} - \omega_{\mu} + \imag ( 1 - \lambda ) \beta ))} 
\\ \cdot
\prod_{\delta = \pm} \prod_{j=1}^{N} \frac{s(\delta x_{j} + \half \omega_{\nu} + \imag \beta / 2 - \imag \lambda \beta )}{s(\delta x_{j} + \half \omega_{\nu} + \imag \beta / 2)} \prod_{k=1}^{\tN} \frac{s(\delta \tx_{k} + \half \omega_{\nu} - \imag \lambda \beta / 2 + \imag \beta )}{s(\delta \tx_{k} + \half \omega_{\nu} - \imag \lambda \beta / 2)},
\end{multline}
and the constant term $c^{0}$ in \eqref{eq_GS_constant}.

The source identity \eqref{eq_Source_identity} then yields the following eigenvalue equation for the CFSV type generalization of the Koornwinder-van Diejen type operators:
\begin{corollary}\label{cor_deformed_vD_GS}
The function 
\begin{equation}\label{eq_KF_deformed_non_deformed}
\Psi_{N, \tN}(\bs x , \bs\tx ; \bs g , \lambda , \beta ) = \frac{\Psi_{N}(\bs x ; \bs g , \lambda ,\beta ) \Psi_{\tN}(\bs \tx ; \bs{g'} , 1 / \lambda , \lambda \beta)}{\prod_{j=1}^{N} \prod_{\delta = \pm} \prod_{k=1}^{\tN} ( s( x_{j} + \delta \tx_{k} + \imag ( \lambda - 1) \beta / 2)s( x_{j} + \delta \tx_{k} - \imag (\lambda - 1) \beta / 2))^{1/2}},
\end{equation}
where $g'_{\nu} = (\lambda + 1 - 2 g_{\nu})/2\lambda$ $(\nu= 0 ,\ldots , 2 \rho +1 )$, satisfies the eigenvalue equation
$$
( H_{N,\tN}(\bs x , \bs \tx ; \bs g , \lambda , \beta) - \cE_{N,\tN}(\bs g , \lambda ,\beta)) \Psi_{N,\tN}(\bs x , \bs \tx ; \bs g , \lambda ,\beta ) = 0
$$
for
$$
\cE_{N,\tN}(\bs g , \lambda ,\beta) = c^{0} + \frac{\prod_{\nu=1}^{\rho}s(\half \omega_{\nu})^{2}}{4} s( \imag \beta [ 2 \lambda N - 2 \tN + \abs{\bs{g}} - \half ( \rho +1)(\lambda +1) ] - \abs{\bs\omega}) 
$$
with $c^{0}$ in \eqref{eq_GS_constant}, in the rational \emph{(I)}, trigonometric \emph{(II)}, and hyperbolic \emph{(III)} cases for all parameters, and for 
\begin{equation}
2 \lambda ( N - 1)  - 2 (\tN + 1) + \sum_{\nu=0}^{7} g_{\nu} = 0 \quad \text{\emph{(IV)}}
\label{eq_deformed_vD_GS_balancing_condition}
\end{equation}
in the elliptic \emph{(IV)} case.
\end{corollary}
Corollary~\ref{cor_deformed_vD_GS} provides the exact groundstate, with corresponding eigenvalue, for the deformed Koornwinder-van Diejen operators. Moreover, the groundstate can be used to construct an explicit weight function for the CFSV type generalization of \eqref{eq_van_Diejen_operator}, and its various limiting cases, obtained from \eqref{eq_lemma_1} using \eqref{eq_deformed_case}:
\begin{corollary}
The operator 
\begin{equation}
A_{N , \tN}(\bs x ,\bs \tx ; \bs g , \lambda ,\beta) = \sum_{\ve=\pm} s(\imag \lambda \beta) \sum_{j=1}^{N} V_{j}^{\ve}(\bs x , \bs\tx) \e^{-\ve \imag \beta \frac{\partial }{\partial x_{j}}} - s(\imag \beta) \sum_{k=1}^{\tN} \tilde{V}_{k}^{\ve}(\bs x , \bs\tx) \e^{+ \ve\imag \lambda \beta \frac{\partial}{\partial \tx_{k}}} + V^{0}(\bs x , \bs \tx)
\label{eq_deformed_operator_reduced}
\end{equation}
has the constant function $1$ as an eigenfunction, with corresponding eigenvalues $\cE_{N, \tN}$, for all parameters in the rational \emph{(I)}, trigonometric \emph{(II)}, hyperbolic \emph{(III)} cases, and for parameters satisfying \eqref{eq_deformed_vD_GS_balancing_condition} in the elliptic \emph{(IV)} case. Furthermore, the operator is formally anti-symmetric w.r.t. a weighted $L^{2}$ inner product with weight function given by
$
\Psi_{N,\tN}(\bs x , \bs\tx ; \bs g ,\lambda , \beta)\Psi_{N,\tN}( - \bs x , - \bs\tx ; \bs g ,\lambda , \beta).
$ 
\end{corollary}

The operator in \eqref{eq_deformed_operator_reduced}, with different normalization, for the rational (I) case originally appeared in the work of Sergeev and Veselov \cite{SV09a}. The hyperbolic version of this operator is due to Feigin and Silantyev \cite{FS14} who also constructed a commuting family of higher order analytic difference operators.

We note that the deformed Koornwinder-van Diejen operators $H_{N,\tN}$, and by extension $A_{N,\tN}$, satisfy
\begin{equation}
\label{eq_deformed_van_Diejen_operator_anti_symmetry}
H_{N, \tN}(\bs x ,\bs \tx ; - \bs g , \lambda , - \beta ) = - H_{N , \tN}(\bs x , \bs \tx ; \bs g , \lambda , \beta)
\end{equation}
and is invariant under the transformation
$$
( N , \tN , \{ g_{\nu}\} , \lambda , \beta) \to ( \tN , N , \{ (2 g_{\nu} - \lambda -1 )/2\lambda \}, 1/\lambda , - \lambda \beta)
$$
and changing the variables correspondingly. Note that anti-symmetry property \eqref{eq_deformed_van_Diejen_operator_anti_symmetry} is due to our normalization, thus changing the normalization allows us to construct an operator which shares the same symmetry as the Koornwinder-van Diejen operators.

Finally, we conclude this section with a kernel function identity for the deformed generalizations of the Koornwinder-van Diejen operators. 
In the special case where $\cN = N + \tN + M + \tM$ ($N,\tN,M,\tM\in\Z_{\gtr0}$) and $\bs X$, $\bs m$ as in \eqref{eq_deformed_KF_case}, we obtain a kernel function for two pairs of deformed Koornwinder-van Diejen type difference operators:
\begin{corollary}\label{cor_deformed_KFI}
The function
\begin{multline}
F_{N,\tN,M,\tM}(\bs x , \bs \tx , \bs y , \bs \ty) = \Psi_{N,\tN}(\bs x , \bs \tx ; \bs g , \lambda , \beta ) \Psi_{M,\tM}(\bs y , \bs \ty ; \bs{\tilde{g}} , \lambda ,\beta )
 \Bigl( \prod_{j=1}^{N}\prod_{k=1}^{M} \prod_{\ve , \ve' = \pm} G( \ve x_{j} + \ve' y_{k} - \imag \lambda \beta / 2 ; \beta ) \Bigr)
\\ \cdot 
 \Bigl( \prod_{j=1}^{N}\prod_{k'=1}^{\tM} \prod_{\delta = \pm} s(x_{j} + \delta \ty_{k'}) \Bigr) 
\Bigl( \prod_{j'=1}^{\tN} \prod_{k=1}^{M} \prod_{\delta = \pm} s(\tx_{j'} + \delta y_{k'})\Bigr) \Bigl( \prod_{j'=1}^{\tN} \prod_{k'=1}^{\tM} \prod_{\ve , \ve' = \pm} G( \ve \tx_{j'} + \ve' \ty_{k'} - \imag \beta / 2 ; \lambda \beta ) \Bigr),
\end{multline}
where $\tilde{g}_{\nu} = \half(\lambda + 1) - g_{\nu}$ $(\nu=0,\ldots,2\rho+1)$, satisfies 
$$
\bigl( H_{N,\tN}(\bs x , \bs \tx ; \bs g , \lambda , \beta ) - H_{M , \tM}(\bs y , \bs \ty ; \bs{\tilde{g}}, \lambda ,\beta) - C_{N,\tN,M,\tM}\bigr) F_{N,\tN,M,\tM}(\bs x , \bs \tx , \bs y , \bs \ty) =0 
$$
with
$$
C_{N,\tN,M,\tM} = \frac{\prod_{\nu=1}^{\rho} s(\half \omega_{\nu})^{2}}{4} s( \imag \beta [ 2 \lambda ( N - M) - 2 (\tN - \tM) + \abs{\bs g} - \half ( \rho + 1) ( \lambda + 1) ] - \abs{\bs\omega}),
$$
for arbitrary parameters in the rational \emph{(I)}, trigonometric \emph{(II)}, and hyperbolic \emph{(III)} cases, and for
$$
2 \lambda ( N - M -1 ) - 2 ( \tN - \tM + 1) + \sum_{\nu=0}^{7} g_{\nu} = 0 \quad \text{\emph{(IV)}}
$$
in the elliptic \emph{(IV)} case.
\end{corollary}

Corollary~\ref{cor_deformed_KFI} is the most general result and all other kernel function identities stated in this section can be obtained as special cases by setting the different variable numbers $N$,$\tN$,$M$, or $\tM$ to zero.

\section{Final remarks}\label{section_final}

In this paper we constructed a Chalykh-Feigin-Sergeev-Veselov type generalization of van Diejen's analytic difference operator and obtained kernel function identities for this operator and its various limiting cases. Using the kernel functions for the deformed Koornwinder-van Diejen type operators it is possible to construct eigenfunctions and eigenvalues of these operators using methods developed in, for example, \cite{Mim01,KNS09,HL10}. In particular, we believe that it is possible to construct generalizations of the Koornwinder polynomials, similar to the generalizations of the Macdonald polynomials in \cite{SV09b}, and extending the results in \cite{SV09a} to the $q$-difference case. In this Section, we outline general properties for these eigenfunctions, obtained from the analytic difference operators and kernel functions. However, a systematical study of the subject is outside the scope of this paper and is left to future work.

Here, we are only considering the trigonometric (II) case and let $z_{j} = \exp(2\imag r x_{j})$ for $j=1,\ldots,N$ and $w_{k} = \exp(2 \imag r \tx_{k})$ for $k=1,\ldots,\tN$. Before proceeding, let us recall that the Weyl group $W_{N} = \Z_{2}^{N} \rtimes \mathfrak{S}_{N}$ of type $BC$ acts naturally on the algebra of Laurent polynomials in $N$ variables through permutation and inversion of the variables. We find that the operator $A_{N,\tN}$ preserves the algebra of Laurent polynomials $p(\bs z , \bs w)$, in the $N+\tN$ variables $\bs z = (z_{1},\ldots,z_{N})$ and $\bs w = (w_{1},\ldots,w_{\tN})$, that are $W_{N}$-invariant in the $\bs z$-variables, $W_{\tN}$-invariant in the $\bs w$-variables, and satisfies the quasi-invariance condition
$$
\Bigl.\bigl(\exp( \imag \half \beta [\frac{\partial}{\partial x_{j}} + \lambda \frac{\partial}{\partial \tx_{k}} ])-\exp( -\imag \half \beta [\frac{\partial}{\partial x_{j}} + \lambda \frac{\partial}{\partial \tx_{k}} ])\bigr) p\Bigr\rvert_{x_{j} = \tx_{k}} = 0
$$
on each hyperplane $x_{j}=\tx_{k}$, for all $j=1,\ldots, N$ and $k=1,\ldots,\tN$. The reader can easily convince themselves that the function obtained when $A_{N,\tN}$ acts on any Laurent polynomial $p$, as described above, only has poles at $z_{j}=0$ for all $j=1,\ldots, N$ and $w_{k} = 0$ for all $k=1,\ldots,\tN$. Furthermore,  it is readily observed from the kernel function that this algebra should be generated by the deformed power sums
$$
\sum_{j=1}^{N} (\e^{2\imag r n x_{j}} + \e^{-2 \imag r n x_{j}}) + \e^{- r n ( \lambda - 1 ) \beta } \frac{1 - \e^{-2r n \beta}}{1- \e^{-2 r  n \lambda \beta}} \sum_{k=1}^{\tN}( \e^{2 \imag r n \tx_{k}} + \e^{-2 \imag r n \tx_{k}} )
$$
for all $n\in\Z_{\gtr0}$.

\section*{Acknowledgments}
We would like to thank E.~Langmann for suggesting this project, and M.~Noumi and S.~ N.~M.~Ruijsenaars for helpful discussions. This work has been supported by the Japan Society for the Promotion of Science and the author is an JSPS International Research Fellow.

\end{document}